\documentclass[
10pt,
superscriptaddress,
nofootinbib,
notitlepage,
amsmath,amssymb,
aps,physrev,
]{revtex4-2}

\usepackage{booktabs}
\usepackage[english]{babel}
\usepackage[T1]{fontenc}
\usepackage{amsmath}
\usepackage{hyperref}
\usepackage[capitalize]{cleveref}
\usepackage{tikz}
\usepackage{lipsum}
\usepackage{amssymb}
\usepackage{amsthm}
\usepackage{enumitem}
\usepackage[labelfont=bf]{caption}

\newtheorem{theorem}{Theorem}

\newtheorem{lemma}{Lemma}

\newtheorem{proposition}{Proposition}
\newtheorem{remark}{Remark}

\newtheorem{assumption}{Assumption}

\crefname{theorem}{Theorem}{Theorems}
\Crefname{theorem}{Theorem}{Theorems}

\crefname{definition}{Definition}{Definitions}
\Crefname{definition}{Definition}{Definitions}

\crefname{lemma}{Lemma}{Lemmas}
\Crefname{lemma}{Lemma}{Lemmas}

\crefname{corollary}{Corollary}{Corollaries}
\Crefname{corollary}{Corollary}{Corollaries}

\crefname{proposition}{Proposition}{Propositions}
\Crefname{proposition}{Proposition}{Propositions}

\crefname{remark}{Remark}{Remarks}
\Crefname{remark}{Remark}{Remarks}

\crefname{problem}{Problem}{Problems}
\Crefname{problem}{Problem}{Problems}

\crefname{assumption}{Assumption}{Assumptions}
\Crefname{assumption}{Assumption}{Assumptions}

\crefname{algocf}{Algorithm}{Algorithms}
\Crefname{algocf}{Algorithm}{Algorithms}
\crefname{algocfline}{Algorithm}{Algorithms}
\Crefname{algocfline}{Algorithm}{Algorithms}

\usepackage{subcaption} 
\usepackage[ruled,vlined,linesnumbered,lined]{algorithm2e}
\usepackage[braket, qm]{qcircuit}
\usepackage{graphicx}
\usepackage{tikz,lipsum,lmodern}

\usepackage{caption}
\usepackage{array}
\usepackage{todonotes}
\usepackage{enumitem} 
\usepackage{makecell}

\begin{document}

\title{Extended Parameter Shift Rules with Minimal Derivative Variance for Parameterized Quantum Circuits}


\author{Zhijian Lai}
\email{lai_zhijian@pku.edu.cn}
\homepage{https://galvinlai.github.io}
\affiliation{Beijing International Center for Mathematical Research, Peking University, Beijing, China}

\author{Jiang Hu}
\email{jianghu@tsinghua.edu.cn}
\homepage{https://hujiangpku.github.io}
\affiliation{Yau Mathematical Sciences Center, Tsinghua University, Beijing, China}

\author{Dong An}
\email{dongan@pku.edu.cn}
\homepage{https://dong-an.github.io}
\affiliation{Beijing International Center for Mathematical Research, Peking University, Beijing, China}

\author{Zaiwen Wen}
\email{wenzw@pku.edu.cn}
\homepage{http://faculty.bicmr.pku.edu.cn/~wenzw}
\affiliation{Beijing International Center for Mathematical Research, Peking University, Beijing, China}

\date{\today}

\begin{abstract}
Parameter shift rules (PSRs) are useful methods for computing arbitrary-order derivatives of the cost function in parameterized quantum circuits. The basic idea of PSRs is to evaluate the cost function at different parameter shifts, then use specific coefficients to combine them linearly to obtain the exact derivatives. In this work, we propose an extended parameter shift rule (EPSR) which generalizes a broad range of existing PSRs and has the following two advantages. First, EPSR offers an infinite number of possible parameter shifts, allowing the selection of the optimal parameter shifts to minimize the final derivative variance and thereby obtaining the more accurate derivative estimates with limited quantum resources. Second, EPSR extends the scope of the PSRs in the sense that EPSR can handle arbitrary Hermitian operator $H$ in gate $U(x) = \exp (iHx)$ in the parameterized quantum circuits, while existing PSRs are valid only for simple Hermitian generators $H$ such as simple Pauli words. Additionally, we show that the widely used ``general PSR'', introduced by Wierichs et al. (2022), is a special case of our EPSR, and we prove that it yields globally optimal shifts for minimizing the derivative variance under the weighted-shot scheme. Finally, through numerical simulations, we demonstrate the effectiveness of EPSR and show that the usage of the optimal parameter shifts indeed leads to more accurate derivative estimates. 
\end{abstract}


\maketitle

\section{Introduction}

Parameterized quantum circuits (PQCs) serve as the foundational components of many hybrid quantum-classical approaches, particularly within variational quantum algorithms (VQAs) \cite{cerezo2021variational, moll2018quantum} and quantum machine learning (QML) \cite{biamonte2017quantum,cerezo2022challenges}. In the VQA framework, the variational quantum eigensolver (VQE) \cite{peruzzo2014variational, kandala2017hardware,grimsley2019adaptive,Shang2023Schrodinger,Yuan2019theory,McArdle2019variational,Endo2020variational} not only efficiently estimates the ground-state energies of molecular systems, but also simulates quantum dynamics, while the quantum approximate optimization algorithm (QAOA) \cite{farhi2014quantum,zhou2020quantum,blekos2024review} offers distinct advantages for tackling combinatorial optimization problems. At the same time, QML uses PQCs for a variety of tasks, including classification, regression, and generative modeling \cite{schuld2019quantum,schuld2020circuit, perdomo2018opportunities,killoran2019continuous,benedetti2019parameterized}. Quantum neural networks (QNNs) \cite{perez2020data,yu2022power,havlivcek2019supervised,abbas2021power}, a prominent subclass of QML models, embed classical data into quantum states or gate spaces and use trainable quantum circuits to learn target functions.

In all these approaches, a quantum circuit is expressed in terms of tunable classical parameters. The circuit is executed on quantum hardware, its outputs are measured, and the results are used to compute a cost function that quantifies performance. A classical optimization algorithm then adjusts the parameters over successive iterations to reduce the cost and enhance the outcome.
Specifically, in this study, we focus on a $q$-qubit quantum system with $N: =2^q$. Without loss of generality, the process of determining the parameters in PQCs essentially translates into solving the following unconstrained optimization problem:
\begin{equation}\label{eq:cost}
    \min _ {\boldsymbol{\theta} \in \mathbb{R}^m} f (\boldsymbol{\theta})=\langle 0|U (\boldsymbol{\theta})^{\dagger} M U (\boldsymbol{\theta})| 0\rangle. 
\end{equation}
Here, $U (\boldsymbol{\theta}) \in \mathbb{C}^{N \times N}$ denotes a PQC that relies on a set of classical parameters $\boldsymbol{\theta}=\left[ \theta_1, \ldots, \theta_m\right]^{T} \in \mathbb{R}^m$, $\ket{0}$ denotes the $q$-qubit state with all the qubits initialized at $0$, and $M \in \mathbb{C}^{N \times N}$ is a Hermitian operator. 
From quantum perspective, $f (\boldsymbol{\theta})$ is the expectation value of the observable $M$, measured with respect to the output state obtained by applying $U (\boldsymbol{\theta})$ onto $\ket{0}$. 
Consistent with numerous studies \cite{sweke2020stochastic, wierichs2022general, mari2021estimating, ding2024random}, we adopt the typical PQC structure as
\begin{equation}\label{eq-Utheta}
    U (\boldsymbol{\theta})=V_m U_m\left (\theta_m\right) \cdots V_1 U_1\left (\theta_1\right),
\end{equation}
where $V_j$'s are composed of fixed quantum gates independently of $\boldsymbol{\theta}$, and $U_j\left (\theta_j\right)$'s are rotation-like gates, defined by
\begin{equation}\label{eq-eiHtheta}
    U_j\left (\theta_j\right)=e^{i H_j \theta_j }, \quad j = 1, \ldots, m,
\end{equation}
for certain Hermitian generators $H_j \in \mathbb{C}^{N \times N}$. 
It should be noted that each $U_j$ is a single-parameter unitary and fully encapsulates the dependence on the univariate parameter $\theta_j \in \mathbb{R}$.

In the context of PQCs, classical optimization techniques play a crucial role.
Broadly, they can be categorized into three main classes: derivative-free methods, gradient-based methods, and structure optimization methods.
Derivative-free methods, such as COBYLA \cite{powell1994direct}, Nelder–Mead \cite{nelder1965simplex}, Powell \cite{powell1964efficient}, and SPSA \cite{spall2000adaptive}, update parameters through direct search, avoiding the need for explicit gradient information.
In contrast, gradient-based methods have attracted increasing attention owing to their advantages, including convergence guarantees and a rich set of well-developed algorithms, such as BFGS \cite{nocedal2006numerical}, L-BFGS \cite{byrd1995limited}, ADAM \cite{kingma2014adam}, AMSGrad \cite{reddi2019convergence}, quantum natural gradient \cite{stokes2020quantum}, and RCD \cite{ding2024random}.
Structure optimization methods \cite{ostaszewski2021structure,nakanishi2020sequential,jager2024fast,lai2025optimal} reconstruct the cost function with respect to single parameter $\theta_j$ via interpolation and subsequently perform a global argmin update on a classical computer.
In gradient-based methods, one typically needs to evaluate the full gradient $\nabla f(\boldsymbol{\theta})$ at each iteration, or at least certain partial derivatives $\partial_j f(\boldsymbol{\theta})$. Here, our interest lies in how such evaluations can be carried out on quantum hardware.

Due to the intrinsic characteristics of quantum circuits, the core technique in neural network training, namely automatic differentiation \cite{baydin2018automatic} and in particular the backpropagation algorithm \cite{rumelhart1986learning}, is difficult to implement on real quantum hardware.
Thus, various methods have been developed for gradient computation in PQCs.
The simplest is the finite difference (FD) method, but FD suffers from numerical instability and inaccuracy.
The Hadamard test (HT) method \cite{guerreschi2017practical,li2024efficient} enables gradient estimation by means of indirect measurements implemented via the Hadamard test; however, it requires additional ancillary qubits and frequent insertion of specific multi-qubit controlled gates into the original circuit, imposing greater demands on hardware implementation.
The parameter shift rule (PSR) \cite{crooks2019gradients,mari2021estimating,kyriienko2021generalized,wierichs2022general,hai2023lagrange,hoch2025variational} offers an efficient approach for computing analytical gradients without the need for ancillary qubits or circuit modifications.
PSR enables exact estimation of derivatives of arbitrary order by evaluating the cost function at multiple finitely shifted parameter configurations and forming an appropriate linear combination of the results.
Although formally similar to the FD method, PSR is exact rather than approximate.
Consequently, it has become the default choice for gradient evaluation in a broad range of PQCs, providing a solid foundation for gradient-based optimization methods.

However, the most frequently cited limitation of PSRs is that they can only be applied to circuits with very specific structures, e.g., where gate generator $H_j $ is simply a Pauli word, or $H_j = \sum \pm P_k$ for some commuting Pauli words $P_k$. This limitation forces practitioners to resort to methods such as FD or HT when dealing with complex circuits. The main contribution of this work is to extend PSR to the broadest possible class of circuits, requiring only that $H_j$ be a Hermitian operator, as is the case for the FD and HT methods.
To summarize, our work makes the following contributions:
\begin{enumerate}

  \item \textbf{Introduction of EPSR.} We propose a novel extension of the standard PSRs, referred to as the extended parameter shift rule (EPSR) and formally stated in \cref{thm:main-psr}, which enables the computation of arbitrary-order derivatives of the cost function in PQCs. Compared with existing PSRs, our EPSR offers the following advantages:
  \begin{enumerate}
      \item EPSR generalizes the many existing PSRs and can handle any Hermitian operator $H_j$, thereby broadening the applicability of PSR in practical quantum applications. In \cref{sec-examples}, we will show in detail that several important PSRs arise as special cases of our EPSR.
    \item The existing PSRs only provide a specific set of parameter shift values and therefore do not consider the question of their optimality. In contrast, EPSR offers an infinite family of possible shift values and determines the corresponding linear combination coefficients, enabling the selection of shifts that minimize derivative variance so as to mitigate statistical noise in quantum measurements. Thus, it allows for more accurate derivative estimates under a fixed total number of shots.
\end{enumerate}
  \item \textbf{Global optimality of equidistant PSR for derivative variance minimization.} 
The so-called ``general PSR'', hereafter referred to as the ``equidistant PSR'' to highlight its essential characteristics, adopts shift values $x_i = \frac{\pi}{2 r} + (i-1) \frac{\pi}{r}$ ($i=1,\dots,2r$) as proposed in \cite{wierichs2022general} (also see \cref{eq:equidistant-shifts-1st,eq-general-eq-psr-1st}). This widely used type of PSR will later be shown to be a specific instance of our EPSR.
From the perspective of derivative variance minimization, we offer a new theoretical interpretation of the equidistant PSR.
Specifically, we rigorously prove that the equidistant PSR under the weighted-shot scheme is the globally optimal parameter shift for minimizing derivative variance. Here, the ``weighted-shot scheme'' allocates the number of shots in proportion to the magnitudes of the linear combination coefficients associated with different shifts, thereby achieving the best allocation under a fixed total number of shots. 
In contrast, most users typically use a constant number of shots for different shifts in their code implementations. We show that, in this case, the equidistant PSR does not achieve the minimum derivative variance, thereby revealing a common and easily made mistake. 

\end{enumerate}

It is worth noting that the above considerations regarding the minimization of derivative variance (i.e., the analysis of optimal shifts) are based on the constant-variance assumption, meaning that the variance of a single measurement is assumed to be independent of the shift value. 
This assumption is commonly adopted in the literature \cite{mari2021estimating,wierichs2022general,markovich2024parameter} to simplify the problem, despite not holding in practice.

We provide both empirical and theoretical support for this constant-variance assumption.
In \cref{sec-experiments}, using the Qiskit simulator, we conduct numerical experiments to validate the conclusions derived under this assumption. The results show that, although the assumption is not strictly satisfied, it remains reasonable in practice, and that employing the optimal parameter shifts in EPSR indeed yields more accurate derivative estimates.
In \cref{app-error-bounds}, we rigorously characterize the discrepancy between the variance minimized under the constant-variance assumption and the true minimal variance when this assumption is removed. The resulting bounds offer qualified theoretical justification for adopting the constant-variance assumption within our framework.

\subsection{Organization}
The rest of this paper is organized as follows. 
In \cref{sec:EPSR}, we present the EPSR as the main result of this work. We then demonstrate that EPSR encompasses a wide range of existing PSR approaches, particularly the equidistant PSR \cite{wierichs2022general}.
In \cref{sec:Derivations}, we provide the derivation of the EPSR.
In \cref{sec:optimal}, we discuss the optimal shifts (termed ``nodes'' in our context) that minimize the derivatives variance. 
In \cref{sec-experiments}, we present numerical experiments applying EPSR to the XXZ model with HVA circuit to validate the conclusions drawn in earlier sections.
Finally, we conclude the paper in \cref{sec-discussion} with a summary and potential directions for future work.

\section{Extended parameter shift rules}\label{sec:EPSR}

In this section, we will propose the extended parameter shift rules. 
We start with an important characteristic of the cost function in \cref{eq:cost} that it can be represented as a trigonometric polynomial, which is crucial for the development of EPSR. 

\subsection{Trigonometric representation}

Consider a parameter vector $\boldsymbol{\theta} \in \mathbb{R}^m$, where all components are fixed except for $\theta_j$ ($j = 1, \dots, m$). When computing the partial derivative for single variable $\theta_j$, operations unrelated to $\theta_j$ can be incorporated into the input state and the observable. This results in the following univariate cost function:
\begin{equation}\label{eq-f-thetaj}
    \theta_j \mapsto f\left (\theta_j\right)
    =\langle\psi|U_j (\theta_j)^{\dagger} C U_j (\theta_j)| \psi\rangle,
\end{equation}
where $|\psi\rangle: = V_{j-1} U_{j-1}\left (\theta_{j-1}\right) \cdots V_1 U_1\left (\theta_1\right) |0\rangle$ represents the state prepared by the subcircuit preceding $U_j\left (\theta_j\right)$, and $C: =V_j^{\dagger} \cdots U_m\left (\theta_m\right)^{\dagger} V_m^{\dagger} M V_m U_m\left (\theta_m\right) \cdots V_j$ is a temporary observable obtained by conjugating $M$ with the subcircuit following $U_j\left(\theta_j\right)$. To simplify the description, we will use letter $x$ to represent $\theta_j$, omitting the subscript $j$ everywhere. 
Thus, our task reduces to computing the derivative of the following univariate function:
\begin{equation}\label{eq-cost-fun}
    f (x) = \left\langle \psi \middle| U (x)^{\dagger} C U (x) \middle| \psi \right\rangle,
\end{equation}
where $C$ is a given observable, $|\psi\rangle$ denotes a fixed quantum state, and $U (x)= e^{i H x}$. 

The eigenvalues of the generator $H$ in \cref{eq-cost-fun} are denoted by $\{\lambda_l \}_{l=1}^{N}$. 
According to \cite{wierichs2022general}, this cost function can be written as a trigonometric polynomial or finite Fourier series (real version):
\begin{equation}\label{eq:trig}
    f (x) = a_0 + \sum_{k=1}^{r} \left[ a_k \cos (\Omega_k x) + b_k \sin (\Omega_k x) \right],
\end{equation}
where $a_0, a_k$ and $b_k$ are some real coefficients, the \textit{frequencies set} $\{\Omega_k\}_{k=1}^{r}$ consists of all positive differences between eigenvalues of $H$, defined as
\begin{equation}\label{eq-Omega-set}
   \operatorname{Frq}^{+}(H) := \left\{\Omega_k \right\}_{k=1}^{r}= \left\{ | \lambda_j - \lambda_l |: \lambda_j \neq \lambda_l, \; \forall j, l = 1, \ldots, N \right\},
\end{equation}
and $r := |\left\{ \Omega_k \right\}|$. Note that this frequency set is completely determined by generator $H$.

As emphasized in \cite{nemkov2023fourier,fontana2022efficient,okumura2023fourier,stkechly2023connecting}, the PQC cost function in \cref{eq:cost} can be interpreted as a truncated multivariate Fourier series. In its complex representation, it takes the form
\begin{equation}
    f(\boldsymbol{\theta}) = \sum_{\boldsymbol{\omega} \in \mathbb{R}^m,\, \omega_j \in \operatorname{Frq}(H_j)} c_{\boldsymbol{\omega}} e^{i \boldsymbol{\omega} \cdot \boldsymbol{\theta}} 
\end{equation}
where $c_{\boldsymbol{\omega}} \in \mathbb{C}$ are complex coefficients, $\boldsymbol{\omega} \cdot \boldsymbol{\theta} :=\sum_{j=1}^m \omega_j \theta_j$, and $\operatorname{Frq}(H) := \left\{  \lambda_j - \lambda_l : \lambda_j ,\lambda_l \text{ are eigenvalues of } H\right\}$. This fact is also highly relevant to theoretical studies on the expressive power of quantum neural networks (QNNs) \cite{schuld2021effect,yu2024non}.
By focusing on a single parameter and applying Euler formula, this complex representation reduces to a real expression as shown in \cref{eq:trig}. The appendices of \cite{wierichs2022general,lai2025optimal} have the detailed proofs for \cref{eq:trig}.
Incidentally, structured optimization algorithms \cite{ostaszewski2021structure,nakanishi2020sequential,jager2024fast,lai2025optimal} for minimizing $f(\boldsymbol{\theta})$ originate from this trigonometric polynomial structure.

Although our discussion concerns a general $H_j$ in \cref{eq-eiHtheta}, there is a particularly important class with extremely wide applications, namely those whose frequencies are equally spaced. Accordingly, we state the following assumption on the frequency set, to be used when required.

\begin{assumption}[Equidistant frequencies]\label{assm-equi}
For every $H_j$ in \cref{eq-eiHtheta}, we assume the frequencies $\{\Omega_{k}^{j}\}_{k  =1}^{r_j}$ are equidistant, i.e., $\Omega_{k}^{j}=k  \Omega^{j}$ $(k=1,\ldots,r_j)$ for some constant $\Omega^{j}$. Without loss of generality, we further restrict the frequencies to integer values, i.e., $\Omega_k^j=k$ $(k= 1, \ldots, r_j)$, since for $\Omega^j \neq 1$ we can rescale the function argument to achieve $\Omega_k^j=k$. Once the rescaled function is obtained, the original function can be readily recovered.
\end{assumption}

As a widely used variant of PSRs, the equidistant PSR \cite{wierichs2022general} falls entirely within the scope of \cref{assm-equi}; hence, any discussion of the equidistant PSR is made under this assumption.
To reconstruct the even and odd components of the cost function in \cref{eq:trig}, \cite{wierichs2022general} employed Dirichlet kernel interpolation, from which the equidistant PSR was derived.
Building on similar insights, we introduce the EPSR, which is likewise based on interpolation applied to \cref{eq:trig}.
However, the EPSR enjoys a significantly broader range of applicability, as it does not generally require \cref{assm-equi}. 
The detailed connections between the two will be presented in \cref{sec-examples}.

\begin{remark}\label{rem-single-frq}
When we focus on a single parameter $\theta_j$ (or write $x$) and thus omit the subscript or superscript $j$, \cref{assm-equi} just states that $\Omega_k=k$ ($k=1,\ldots,r$) in \cref{eq:trig}.
The simplest equidistant case occurs when $H_j$ is a single Pauli word, in which case the eigenvalues are $\pm 1$, and hence the frequency set is $\{2\}$.
Note, however, that in much of the literature the gate is defined with a prefactor $1 / 2$ in $e^{i H_j \theta_j / 2}$, which results in an actual equidistant frequency set $\{1\}$ in our context.
\end{remark}

\subsection{Main theorem: Extended parameter shift rules}

In the following, $f^{(d)}$ denotes the $d$-th order derivative of some univariate real-valued function $f$. The main result of this paper is presented in the following theorem. The derivations of \cref{thm:main-psr} will be given in \cref{sec:Derivations}. 

\begin{theorem}[Extended parameter shift rules (EPSR)]\label{thm:main-psr}
Let $d \geq 1$ be an integer.
Consider the cost function $f$ defined in \cref{eq-cost-fun}, equivalently represented in \cref{eq:trig}, where the frequency set $\{\Omega_k\}_{k=1}^r$ is specified in \cref{eq-Omega-set} and derived from the Hermitian generator $H.$
Let $\bar{x} \in \mathbb{R}$ be a fixed parameter at which the derivatives will be evaluated. 

1. For odd $d$, we have
\begin{equation}\label{eq-odd}
    f^{ (d)} (\bar{x}) =\frac{1}{2} \sum_{i=1}^{r} b_i (\mathbf{x}) \Bigl [f (\bar{x} + x_i) - f (\bar{x} - x_i)\Bigr],
\end{equation}
where shifts (or, nodes)\footnote{As shown later, the vector $\mathbf{x}$ represents (interpolation) nodes, commonly referred to as ``shifts'' in the PSR literature. We use the terms interchangeably.} vector $\mathbf{x} =[x_1, \dots, x_r]^T \in \mathbb{R}^r$ can be chosen arbitrarily, provided that the $r \times r$ matrix
\begin{equation}\label{eq-Ao}
    A_o (\mathbf{x}): = \begin{bmatrix}
    \sin (\Omega_1 x_1) & \sin (\Omega_2 x_1) & \cdots & \sin (\Omega_r x_1) \\
    \sin (\Omega_1 x_2) & \sin (\Omega_2 x_2) & \cdots & \sin (\Omega_r x_2) \\
    \vdots & \vdots & \ddots & \vdots \\
    \sin (\Omega_1 x_r) & \sin (\Omega_2 x_r) & \cdots & \sin (\Omega_r x_r)
    \end{bmatrix}
\end{equation}
is nonsingular. Here, the coefficient vector $b (\mathbf{x}) \in \mathbb{R}^r$ is determined by solving the linear equation
\begin{equation}
    b (\mathbf{x}) = \Bigl[A_o (\mathbf{x})^T\Bigr]^{-1} p^{ (d)},
\end{equation}
with $p^{ (d)}: = (-1)^{\frac{d-1}{2}}\left[\Omega_1^d, \Omega_2^d, \cdots, \Omega_r^d\right]^T \in \mathbb{R}^r$.

2. For even $d$, we have
\begin{equation}\label{eq-even}
    f^{ (d)} (\bar{x}) = \frac{1}{2}\sum_{i=0}^{r} b_i (\mathbf{x}) \Bigl [f (\bar{x} + x_i) + f (\bar{x} - x_i)\Bigr],
\end{equation}
where shifts vector $\mathbf{x} =[x_0, x_1, \dots, x_r]^T \in \mathbb{R}^{r+1}$ can be chosen arbitrarily, provided that the $(r+1) \times (r+1)$ matrix
\begin{equation}\label{eq-Ae}
    A_e (\mathbf{x}): = \begin{bmatrix}
    1 & \cos (\Omega_1 x_0) & \cos (\Omega_2 x_0) & \cdots & \cos (\Omega_r x_0) \\
    1 & \cos (\Omega_1 x_1) & \cos (\Omega_2 x_1) & \cdots & \cos (\Omega_r x_1) \\
    \vdots & \vdots & \vdots & \ddots & \vdots \\
    1 & \cos (\Omega_1 x_r) & \cos (\Omega_2 x_r) & \cdots & \cos (\Omega_r x_r)
    \end{bmatrix}
\end{equation}
is nonsingular. Here, the coefficient vector $b (\mathbf{x}) \in \mathbb{R}^{r+1}$ is determined by solving the linear equation
\begin{equation}
    b (\mathbf{x}) = \Bigl[A_e (\mathbf{x})^T\Bigr]^{-1} q^{ (d)},
\end{equation}
with $q^{ (d)}: = (-1)^{\frac{d}{2}}\left[\delta_{d, 0}, \Omega_1^d, \Omega_2^d, \cdots, \Omega_r^d\right]^T \in \mathbb{R}^{r+1}$. Here, Kronecker delta $\delta_{d, 0}$ equals 1 if $d$ is 0, and 0 otherwise.
\end{theorem}

\subsubsection{Illustrative examples of \cref{thm:main-psr}}\label{sec-examples}

To better understand \cref{thm:main-psr}, we show that it covers a broad class of existing approaches on PSR \cite{li2017hybrid,mitarai2018quantum,schuld2019evaluating,anselmetti2021local,wierichs2022general}. 

\paragraph{Example I: frequency \{1\}.}
When order $d = 1$ and the generator $H$ has the singleton frequency $\{1\}$ (i.e., $H$ has only two distinct eigenvalues, see \cref{rem-single-frq}), we have $r= 1$, $A_o = \sin(x_1)$, $p^{(1)} = 1$, and $b = 1 / \sin(x_1)$. In this case, \cref{eq-odd} reduces to
\begin{equation}
    f'(\bar{x}) = \frac{1}{2 \sin(x_1)} \left[ f(\bar{x} + x_1) - f(\bar{x} - x_1) \right],
\end{equation}
where $x_1 \in (0, \pi)$. This is exactly the early version of PSR presented in \cite{li2017hybrid,mitarai2018quantum,schuld2019evaluating}.

\paragraph{Example II: frequency \{1,2\}.}
When order $d=1$ and the generator $H$ has the eigenvalues $\{-1,0,1\}$, resulting in $r=2$ and frequencies $\Omega_1=1,$ $\Omega_2=2$, we have $p^{(1)}= [1,2]^{T}$ and
\begin{equation}
    A_o(\mathbf{x}) = \begin{bmatrix}
    \sin \left(x_1\right) & \sin \left(2 x_1\right) \\
    \sin \left(x_2\right) & \sin \left(2 x_2\right)
    \end{bmatrix}.
\end{equation}
If we choose $\mathbf{x}=[\frac{\pi}{4}, \frac{3\pi}{4}]^{T}$, then $b = [b_1, b_2]^{T} = \left(A_o(\mathbf{x})^{T}\right)^{-1} p^{(1)}= \left[ (1+\sqrt{2})/\sqrt{2} , (1-\sqrt{2})/\sqrt{2} \right]^{T}.$ Thus, \cref{eq-odd} becomes
\begin{align}
    f^{\prime}(\bar{x}) 
    = \frac{1}{2} b_1 \left[f\left(\bar{x}+x_1\right) - f\left(\bar{x}-x_1\right)\right] + \frac{1}{2}  b_2 \left[f\left(\bar{x}+x_2\right) - f\left(\bar{x}-x_2\right)\right].
\end{align}
This special result is exactly the PSR studied in \cite{anselmetti2021local}.

\paragraph{Example III: frequency \{1,2, \ldots, r\}.}
For order $d = 1,2$, the equidistant PSR \cite{wierichs2022general} was analyzed under \cref{assm-equi} where $\Omega_k = k$ for $k = 1, \ldots, r$.
For $d=1$, equidistant PSR introduced the equidistant PSR nodes
\begin{equation}\label{eq:equidistant-shifts-1st}
    x_i = \frac{\pi}{2 r} + (i-1) \frac{\pi}{r}, \qquad i = 1, \ldots, 2r,
\end{equation}
to obtain
\begin{equation}\label{eq-general-eq-psr-1st}
f^{\prime}(\bar{x})=\sum_{i=1}^{2 r} c_{i} \, f\left(\bar{x}+x_i\right), \quad \text{ with } \quad  c_i = \frac{(-1)^{i-1}}{4 r \sin ^2\left(\frac{1}{2} x_i \right)}.
\end{equation}
For $d=2$, equidistant PSR \cite{wierichs2022general} also introduced the equidistant PSR nodes
\begin{equation}\label{eq:equidistant-shifts-2rd}
    x_i = \frac{i\pi}{r}, \qquad i = 0,1, \ldots, 2r-1,
\end{equation}
to obtain second order derivative
\begin{equation}\label{eq-general-eq-psr-2rd}
    f^{\prime\prime}(\bar{x})
    =
    -\frac{2 r^2+1}{6}f(\bar{x}) + \sum_{i=1}^{2 r-1} c_{i} \, f\left(\bar{x}+x_i\right), \quad \text{ with } \quad  c_i = \frac{(-1)^{i-1}}{2 r \sin ^2\left(\frac{1}{2} x_i \right)}.
\end{equation}
Indeed, when we apply odd argument in \cref{thm:main-psr} by setting the $i$-th component of shift $\mathbf{x}\in \mathbb{R}^r$ as $x_i$ for $i=1,\ldots,r$ in  \cref{eq:equidistant-shifts-1st}, we can compute the corresponding coefficients $b^{\prime}:=\frac{1}{2}\left[b(\mathbf{x})^T; -b(\mathbf{x})^T\right]^T$ in \cref{eq-odd}, which are exactly the coefficients $c_i$ in \cref{eq-general-eq-psr-1st}. Similarly, when we apply even argument in \cref{thm:main-psr} by setting the $i$-th component of $\mathbf{x}\in \mathbb{R}^{r+1}$ as $x_i$ for $i=0,1,\ldots,r$ in  \cref{eq:equidistant-shifts-2rd}, we also have the same coefficients $c_i$ in \cref{eq-general-eq-psr-2rd}. 
The rigorous proof can be found in \cref{app-equ-are-special}.
In other words, \cref{eq-general-eq-psr-1st,eq-general-eq-psr-2rd} provide the analytical forms of the coefficients $b$ under the specific cases of choosing $\mathbf{x}$.
Therefore, the equidistant PSRs \cite{wierichs2022general} are instances of \cref{thm:main-psr}.

\begin{remark}
In \cref{subsec:not-optimal} later, we will show that when a weighted-shot scheme is adopted, equidistant PSR are globally optimal to obtain the minimal variance of estimated derivatives. Such optimality is not discussed in the original literature.
\end{remark}

\subsubsection{More discussion of \cref{thm:main-psr}}

Let us return to the statement of \cref{thm:main-psr}.
For general values of $\Omega_k$, it is challenging to characterize the conditions under which a given choice of $\mathbf{x}$ renders the matrix $A_o(\mathbf{x})$ or $A_e (\mathbf{x})$ nonsingular. However, in the equidistant frequencies case (\cref{assm-equi}), which commonly arises when $H$ is simply a Pauli word, or $H= \sum \pm P_k$ for commuting Pauli words $P_k$, we summarize the precise conditions as following theorem.
The proof can be found in \cref{app-thm-cond}.

\begin{theorem}\label{thm-cond-nonsingular}
Let \cref{assm-equi} hold, i.e., $\Omega_k = k$. Then,
\begin{enumerate}
    \item the matrix $A_o (\mathbf{x})$ with $\mathbf{x}\in \mathbb{R}^r$ is nonsingular, if and only if
\begin{itemize}
    \item all $\sin x_i \neq 0 \Leftrightarrow x_i \notin \pi \mathbb{Z}$ for all $i$;
    \item all $\cos x_i$ are pairwise distinct $\Leftrightarrow$ $x_i \not \equiv \pm x_j \pmod{2\pi}$ for all $i \neq j$.
\end{itemize}
\item the matrix $A_e (\mathbf{x})$ with $\mathbf{x}\in \mathbb{R}^{r+1}$ is nonsingular, if and only if
\begin{itemize}
    \item all $\cos x_i$ are pairwise distinct $\Leftrightarrow$ $x_i \not \equiv \pm x_j \pmod{2\pi}$ for all $i \neq j$.
\end{itemize}
\end{enumerate}
\end{theorem}

\cref{thm:main-psr} states that evaluating an odd-order derivative requires at most $2r$ function evaluations of $f$, while evaluating an even-order requires at most $2r+2$ evaluations. 
However, for even-order, this can be reduced to $2r+1$ by setting $x_0 = 0$ in \cref{eq-even}, and $A_e$ may still be nonsingular. 
(In contrast, for odd-order, choosing any $x_i = 0$ would make the $i$-th row of $A_o$ in \cref{eq-Ao} all-zero, rendering $A_o$ singular.)
Furthermore, if \cref{assm-equi} holds, then $f$ becomes $2\pi$-periodic, implying $f(\bar{x}+\pi) = f(\bar{x}-\pi)$. In this case, for even-order, setting the last node $x_r = \pi$ allows an additional reduction of function evaluations in \cref{eq-even}, ultimately requiring only $2r$ calls to $f$.  
(Again, for odd-order, choosing any $x_i = \pi$ would also make $A_o$ singular, see \cref{thm-cond-nonsingular}.)
Since equidistant frequencies are the vary common in practice, this means that in many cases, regardless of the order of the derivative, only $2r$ function evaluations are needed.

We observe that to obtain the coefficient vector $b(\mathbf{x})$, we need to solve a linear equation $A_o (\mathbf{x})^T b (\mathbf{x})= p^{ (d)}$ (or $A_e (\mathbf{x})^T b (\mathbf{x})= q^{ (d)}$), where the order $d$ only affects the constant vector $p^{ (d)}$ (or $q^{ (d)}$) on the right-hand side. Therefore, once the inverse of $A_o (\mathbf{x})$ (or $A_e (\mathbf{x})$) is computed, we can calculate the PSR for any odd (or even) order. Since the shift $\mathbf{x}$ for function evaluation does not need to change, the function evaluation values $f\left (\bar{x} \pm x_i \right)$ can be reused. Accordingly, in theory, once $f^{ (1)} (\bar{x})$ has been computed, any higher-order odd-order derivatives can be obtained without further evaluations of $f$; the only step needed is to recompute the coefficients $b(\mathbf{x})$. 
Furthermore, notice that the linear equation $A_o (\mathbf{x})^T b (\mathbf{x})= p^{ (d)}$ (or $A_e (\mathbf{x})^T b (\mathbf{x})= q^{ (d)}$) does not depend on which point the derivatives are evaluated at (i.e., the value of $\bar{x}$). When we apply our EPSR in training of PQCs, we only need to choose the nodes and solve the linear equation once before we perform training, and the nodes and coefficients can remain unchanged during the entire iterations.

On the other hand, \cref{thm:main-psr} tells us that the shift $\mathbf{x}$ can be chosen freely as long as $A_o(\mathbf{x})^{-1}$ (or $A_e(\mathbf{x})^{-1}$) exists. 
However, in practice, the function evaluations for $f$ are only obtained on a quantum device through a finite number of measurements, followed by averaging, which inherently introduces noise into the true function values. In \cref{sec:optimal}, we will discuss how to optimally choose the shift $\mathbf{x}$ to ensure that the PSR performs robustly against noise. 
This represents a key distinction between our work and existing studies on the PSRs: prior works either focus on specific choices of shift $\mathbf{x}$ and their corresponding coefficients $b$, or do not address how to select $\mathbf{x}$ to minimize the final variance of the derivative estimate. We will explore these issues in detail in the following sections.

\subsection{Comparison with a similar study}\label{sec-compare}

We now compare \cref{thm:main-psr} with the related results in \cite{markovich2024parameter}, which recently analyzed the same cost function \cref{eq-cost-fun} and obtained a theorem of a similar form. Suppose that in \cref{eq-cost-fun}, the $N \times N$ generator $H$ has $N$ eigenvalues $\lambda_k$, ordered as $\lambda_1 \leq \cdots \leq \lambda_N$.  \cite{markovich2024parameter} proposed the following PSR: for any $t \in \mathbb{R},$
\begin{equation}\label{eq-mark-psr}
    f^{\prime}(t)=\sum_{i=1}^m b_i(\vec{\phi}) f\left(t+\phi_i\right),
\end{equation}
where the shift vector $\vec{\phi} \in \mathbb{R}^m$ is freely chosen with $m :=N(N-1)+1$, and the coefficient vector $b(\vec{\phi})$ is the solution to the linear equation $E(\vec{\phi}) b = \vec{\mu}$. Those notations are provided as follows: the difference between any two eigenvalues is defined as 
$\mu_{(k, l)} := \lambda_l - \lambda_k, \forall k, l=1,\ldots,N.$
Note that $\mu_{(k, l)}=-\mu_{(l, k)}$ and $\mu_{(k, k)}=0$. Then, let $\vec{\mu}$ denote the $m \times 1$ vector with elements $i \mu_{(k, l)}$, i.e.,
\begin{equation}
    \vec{\mu} :=\left(\begin{array}{c}
    0 \\
    \phantom{+}i \mu_{(1,2)} \\
    -i \mu_{(1,2)} \\
    \vdots \\
    \phantom{+}i \mu_{(N-1,N)} \\
    -i \mu_{(N-1,N)}
    \end{array}\right).
\end{equation}
Here, the pair $(k, l)$ is re-indexed from 1 to $m$, and the number $m$ accounts for all $\mu_{(k, l)}$ where $k \neq l$ for $k=1, \ldots, N$, plus only one instance of $\mu_{(k, k)}=0$. The $m \times m$ matrix $E(\vec{\phi})$ is given by
\begin{align}
    E(\vec{\phi}) 
    :=  
     \left[\begin{array}{cccc}
    1 & 1 & \ldots & 1 \\
    e^{i \mu_{(1,2)} \phi_1} & e^{i \mu_{(1,2)} \phi_2} & \ldots & e^{i \mu_{(1,2)} \phi_m} \\
    e^{-i \mu_{(1,2)} \phi_1} & e^{-i \mu_{(1,2)} \phi_2} & \ldots & e^{-i \mu_{(1,2)} \phi_m} \\
    \vdots & \vdots & \ddots & \vdots \\
    e^{i \mu_{(N-1, N)} \phi_1} & e^{i \mu_{(N-1, N)} \phi_2} & \ldots & e^{i \mu_{(N-1, N)} \phi_m} \\
    e^{-i \mu_{(N-1, N)} \phi_1} & e^{-i \mu_{(N-1, N)} \phi_2} & \ldots & e^{-i \mu_{(N-1, N)} \phi_m}
    \end{array}\right]. \label{eq-Ephi}
\end{align}
The form of \cref{eq-mark-psr} shares some similarities with ours in \cref{eq-odd}. We both need to solve a linear equation that depends on shift vector $\phi$ (i.e., $\mathbf{x}$ in our notations). But the design of two equations is very different. Unfortunately, compared to \cref{thm:main-psr}, the following configurations of eigenvalues of generator $H$ can cause the equation $E(\vec{\phi}) b(\vec{\phi}) = \vec{\mu}$ to be ill-posed, making the PSR in \cref{eq-mark-psr} invalid.
\begin{enumerate}
    \item When the eigenvalues of $H$ are equidistant. Here, ``equidistant'' means that the difference between every two consecutive eigenvalues (arranged in non-decreasing order) remains constant, denoted by $\Delta$. For example, if $\mu_{(1,2)}= \mu_{(2,3)}=\Delta$, then there are two identical rows in \cref{eq-Ephi} and $E(\vec{\phi})$ becomes singular. Although \cite{markovich2024parameter} acknowledged this issue and modified $E(\vec{\phi})$ and $\mu$ accordingly, the resulting conclusion is that when $H$ has equidistant eigenvalues, only one particular equidistant shift $\phi$ ensures the validity of \cref{eq-mark-psr}. As a comparison, regardless of whether the eigenvalues of $H$ are equidistant, no additional modifications are required in our \cref{thm:main-psr}, and the shift vector can always be freely chosen.
    \item When any two differences $\mu_{(k, l)}$ are equal (the eigenvalues themselves do not necessarily have to be equidistant).  For example, if $\mu_{(1,2)} = \mu_{(N-1,N)}$, then again two rows in \cref{eq-Ephi} are identical so $E(\vec{\phi})$ becomes singular. Hence, applying \cref{eq-mark-psr} requires that all differences between every pair of eigenvalues of $H$ must be distinct. 
    \item When $H$ has degenerate eigenvalues. Then some $\mu_{(k, l)}$ with $k \neq l$ will be zero, which would cause $E(\vec{\phi})$ to be singular. 
\end{enumerate}

We can see that \cref{thm:main-psr} not only handles the aforementioned three situations of eigenvalue distributions, but also has the following computational cost advantages\footnote{In the following comparison, assume that we compute the first-order derivative and $H$ has $n$ eigenvalues such that their distribution makes $E(\vec{\phi})$ nonsingular.}:

\begin{enumerate}
    \item Our method requires fewer function evaluations compared to \cref{eq-mark-psr}. While \cref{eq-mark-psr} requires a fixed $m = N(N-1) + 1$ function evaluations, in our method, $r$ in \cref{eq-Omega-set} is at most $N(N-1)/2$, resulting in a total of $2r = N(N-1)$ function evaluations. Thus, the upper bound on evaluations in our method is one less than that of \cref{eq-mark-psr}.
    \item The size of the linear system to be solved in \cref{thm:main-psr} is significantly smaller than that in \cref{eq-mark-psr}. The system in \cref{eq-Ephi} has size $m = N(N-1) + 1$ and consists of complex elements, whereas in our method, the size of \cref{eq-Ao-full} is $r$, with $r$ being at most $N(N-1)/2$. Consequently, our system is at least half the size of \cref{eq-mark-psr} and moreover it involves only real entries.
\end{enumerate}

In summary, although both \cref{thm:main-psr} and the result in \cite{markovich2024parameter} can handle cases where the eigenvalues of $H$ are not equally spaced, our \cref{thm:main-psr} is more robust, capable of addressing a wider range of scenarios, and incurs relatively lower computational costs.

\section{Mathematical derivation of the EPSR} \label{sec:Derivations}

In this section, we derive our EPSR stated in \cref{thm:main-psr}.
The derivation here is natural and intuitive, relying on no advanced mathematical tools.
We begin by verifying \cref{thm:main-psr} at $\bar{x} = 0$.
To this end, we decompose the cost function $f$ into its odd and even components, and express their $d$-th derivatives using an inner product representation.
This analysis shows that the odd-order derivatives of $f$ at $\bar{x} = 0$ depend solely on the odd component, whereas the even-order derivatives depend solely on the even component.
Next, we employ an interpolation method to recover the coefficients, carefully choosing the interpolation nodes and constructing an appropriate interpolation matrix.
These derivations show that the derivatives of $f$ at zero can be expressed as a weighted sum of its function values.
Finally, by shifting the evaluation point, we extend this result to compute derivatives at any desired $\bar{x}$, not just at zero.

Let us first discuss some auxiliary results. Given a smooth function $f: \mathbb{R} \to \mathbb{R}$, we define its odd and even parts as
\begin{equation}
    f_o (x): =\frac{1}{2}[f (x)-f (-x)], \quad 
    f_e (x): =\frac{1}{2}[f (x)+f (-x)].
\end{equation}
Clearly, $f (x)=f_o (x)+f_e (x)$. Indeed, for any $d \geq 0$, their $d$-th order derivatives have the same results, $f^{ (d)} (x)=f_o^{ (d)} (x)+f_e^{ (d)} (x)$. 
In what follows, we are especially interested in the case of $x=0$, i.e., $f^{ (d)} (0)=f_o^{ (d)} (0)+f_e^{ (d)} (0).$
The following lemma is straightforward. 
\begin{lemma}\label{lem:shifted-f}
Let $f: \mathbb{R} \to \mathbb{R}$ be a smooth function, and define its shifted version $h(x) := f(x + \bar{x})$ for some fixed $\bar{x} \in \mathbb{R}$. Then for any $d$-th derivative, we have $h^{(d)}(0) = f^{(d)}(\bar{x}).$
\end{lemma}

\subsection{Inner product representations}

Consider the cost function $f$ in \cref{eq:trig}. Then, its odd part and even part are given by
\begin{equation}\label{eq-odd-part}
    f_o(x)= \sum_{ k =1}^{r} b_{k} \sin \left (\Omega_{k} x\right),
\end{equation}
and
\begin{equation}\label{eq-even-part}
    f_e(x)=a_0+\sum_{ k =1}^{r} a_{k} \cos \left (\Omega_{k} x\right).
\end{equation}
We first rewrite these functions using the notation of the Frobenius inner product $\langle A, B\rangle =\operatorname{tr}(A^T B)$ for any two matrices or vectors $A$ and $B$ of the same sizes. Later, we often use the properties $\langle C A, B\rangle =\left\langle A, C^{T} B\right\rangle, \langle AD, B\rangle =\langle A, BD^{T} \rangle$ with matrices $A,B,C,D$ of compatible sizes.

\subsubsection{Reformulation of the odd part}

Consider the odd part \cref{eq-odd-part}. For any nonnegative integer $d$, using inner product natation, we have
\begin{equation}\label{eq-reform-odd}
    f_o^{ (d)}(x)= \sum_{ k =1}^{r} b_{k} \sin ^{ (d)}\left (\Omega_{k} x\right)
    =\langle p^{ (d)} (x), z_o\rangle,
\end{equation}
where we define the odd coefficient vectors
\begin{equation}
    z_o: =\left[b_1, \cdots, b_r\right]^{T} \in \mathbb{R}^r
\end{equation}
and $p^{ (d)} (x): =[\sin ^{ (d)} (\Omega_{1} x), \ldots, \sin ^{ (d)} (\Omega_{r} x)]^{T} \in \mathbb{R}^r.$ If we expand $p^{ (d)} (x)$ specifically, then
\begin{align}
    p^{ (d)} (x) =
    \begin{cases}
    +[\Omega_1^d \sin (\Omega_1 x), \ldots, \Omega_r^d \sin (\Omega_r x)]^{T}, & \text{if } d = 0 \pmod{4}, \\
    +[\Omega_1^d \cos (\Omega_1 x), \ldots, \Omega_r^d \cos (\Omega_r x)]^{T}, & \text{if } d = 1 \pmod{4}, \\
    -[\Omega_1^d \sin (\Omega_1 x), \ldots, \Omega_r^d \sin (\Omega_r x)]^{T}, & \text{if } d = 2 \pmod{4}, \\
    -[\Omega_1^d \cos (\Omega_1 x), \ldots, \Omega_r^d \cos (\Omega_r x)]^{T}, & \text{if } d = 3 \pmod{4}.
    \end{cases}
\end{align}
This expansion leverages the periodicity of the higher-order derivatives of the sine function.
When $d$ is even, then $p^{(d)}(0)$ becomes zero vector, which implies that $f_o^{(d)}(0) = 0$ by \cref{eq-reform-odd}. When $d$ is odd, we define
\begin{equation}\label{eq-pd0}
    p^{(d)} \equiv p^{(d)}(0) = (-1)^{\frac{d-1}{2}} [\Omega_1^d, \Omega_2^d, \cdots, \Omega_r^d]^T,
\end{equation}
where the factor $(-1)^{\frac{d-1}{2}}$ equals $1$ if $d = 1 \pmod{4}$ and $-1$ if $d = 3 \pmod{4}$.

\subsubsection{Reformulation of the even part}

We next consider the even part  \cref{eq-even-part}. For any nonnegative integer $d$, similarly, we have
\begin{equation}\label{eq-reform-even}
    f_e^{ (d)}(x)
    = a_0 + \sum_{ k =1}^{r} a_{k} \cos ^{ (d)}\left (\Omega_{k} x\right)
    =\langle q^{ (d)} (x), z_e \rangle,
\end{equation}
where we define the even coefficient vectors
\begin{equation}
    z_e: =\left[a_0, a_1, \cdots, a_r\right]^{T} \in \mathbb{R}^{r+1}
\end{equation}
and $q^{ (d)} (x): =[\delta_{d, 0}, \cos ^{ (d)} (\Omega_{1} x), \ldots, \cos ^{ (d)} (\Omega_{r} x)]^{T} \in \mathbb{R}^{r+1}.$ Here, Kronecker delta $\delta_{d, 0}$ equals 1 if $d=0$, and 0 otherwise. Similarly, expanding $q^{ (d)} (x)$ gives
\begin{align}
   q^{ (d)} (x) = 
   \begin{cases}
    +[\delta_{d, 0}, \Omega_1^d \cos (\Omega_1 x), \ldots, \Omega_r^d \cos (\Omega_r x)]^{T}, & \text{if } d = 0 \pmod{4}, \\
    -[0 \quad, \Omega_1^d \sin (\Omega_1 x), \ldots, \Omega_r^d \sin (\Omega_r x)]^{T}, & \text{if } d = 1 \pmod{4}, \\
    -[0\quad, \Omega_1^d \cos (\Omega_1 x), \ldots, \Omega_r^d \cos (\Omega_r x)]^{T}, &\text{if } d = 2 \pmod{4}, \\
    +[0\quad, \Omega_1^d \sin (\Omega_1 x), \ldots, \Omega_r^d \sin (\Omega_r x)]^{T}, & \text{if } d = 3 \pmod{4}.
    \end{cases}
\end{align}
When $d$ is odd, then $q^{(d)}(0)$ becomes zero vector, thus $f_e^{(d)}(0) = 0$ by \cref{eq-reform-even}. When $d$ is even, we define
\begin{equation}\label{eq-qd0}
    q^{(d)} \equiv q^{(d)}(0) = (-1)^{\frac{d}{2}} [ \delta_{d,0}, \Omega_1^d, \Omega_2^d, \cdots, \Omega_r^d ]^T.
\end{equation}
where the factor $(-1)^{\frac{d}{2}}$ equals $1$ if $d = 0 \pmod{4}$ and $-1$ if $d = 2 \pmod{4}$.

Finally, based on above discussion and $f^{ (d)} (0)=f_o^{ (d)} (0)+f_e^{ (d)} (0)$, we have
\begin{equation}\label{eq-odd-even-d-0}
    f^{(d)}(0)=\begin{cases}
    f_o^{(d)}(0) & \text { if } d \text { is odd, } \\ f_e^{(d)}(0) & \text { if } d \text { is even. }
    \end{cases}
\end{equation}
Consequently, at $\bar{x}=0$, the odd-order derivatives of $f$ are determined entirely by its odd component, while the even-order derivatives are determined solely by its even component.

\subsection{Interpolations and derivatives}

Note that once the coefficient vectors $z_o$ in \cref{eq-reform-odd} and $z_e$ in \cref{eq-reform-even} are determined, the $d$-th derivative of odd and even parts can be computed directly. To this end, we next recover coefficients $z_o$ and $z_e$ by using interpolation methods. 

\subsubsection{Computation of the odd order derivatives}

We again first consider the odd part \cref{eq-odd-part}. We are free to choose $r$ interpolation nodes $\mathbf{x}=[x_1, \cdots, x_r]^{T} \in \mathbb{R}^r$ as long as the  $r \times r$ interpolation matrix
\begin{equation}\label{eq-Ao-full}
    A_o \equiv A_o(\mathbf{x}):=\left[\begin{array}{c}
    p^{(0)}\left(x_1\right)^{T} \\
    p^{(0)}\left(x_2\right)^{T} \\
    \vdots \\
    p^{(0)}\left(x_r\right)^{T}
    \end{array}\right]=\left[\begin{array}{cccc}
    \sin \left(\Omega_1 x_1\right) & \sin \left(\Omega_2x_1\right) & \cdots & \sin \left(\Omega_r x_1\right) \\
    \sin \left(\Omega_1 x_2\right) & \sin \left(\Omega_2 x_2\right) & \cdots & \sin \left(\Omega_r x_2\right) \\
    \vdots & \vdots & \ddots & \vdots \\
    \sin \left(\Omega_1 x_r\right) & \sin \left(\Omega_2 x_r\right) & \cdots & \sin \left(\Omega_r x_r\right)
    \end{array}\right]
\end{equation}
is nonsingular. Define data vector
\begin{equation}
    y_o \equiv y_o (\mathbf{x}): =\left[f_o\left (x_1\right), \cdots, f_o\left (x_r\right)\right]^{T} \in \mathbb{R}^r.
\end{equation}
Then, by definitions of odd part \cref{eq-odd-part}, we have $A_o z_o =y_o$ and thus $z_o =A_o ^{-1} y_o$. Thus, by \cref{eq-reform-odd},
\begin{align}
    f_o ^{(d)}(x)
    =\langle p^{(d)}(x), A_o ^{-1} y_o \rangle
   =\langle(A_o ^{T})^{-1} p^{(d)}(x), y_o \rangle.
\end{align}
In particular, we focus on case that $d$ is odd and $x=0$. By \cref{eq-odd-even-d-0,eq-pd0}, we have
\begin{equation}
    f ^{(d)}(0)
    =f_o ^{(d)}(0)
    =\left\langle b(\mathbf{x}), y_o(\mathbf{x}) \right\rangle,
\end{equation}
with $b(\mathbf{x}):=[A_o(\mathbf{x})^T]^{-1} p^{(d)} \in \mathbb{R}^r$. Hence,
\begin{align}
    f^{(d)}(0) 
    =\sum_{i=1}^r b_i(\mathbf{x}) f_o\left(x_i\right)
    = \frac{1}{2} \sum_{i=1}^r b_i(\mathbf{x})\left(f\left(x_i\right)-f\left(-x_i\right)\right). \label{eq-557}
\end{align}

Now, we have completed the odd argument in \cref{thm:main-psr} for $\bar{x} = 0$. Notice that the coefficient vector $b (\mathbf{x})$ above is generated by solving the linear system
\begin{equation}
    A_o (\mathbf{x})^T b (\mathbf{x}) = p^{ (d)},
\end{equation}
where $A_o (\mathbf{x})$ and $p^{ (d)}$ only depend on the desired order $d$ and the frequency set $\{\Omega_k\}$, but not on the coefficients $a_k$ and $b_k$ in $f$. Consequently, the EPSR of \cref{eq-557} is independent of the coefficients $a_k$ and $b_k$.

To obtain $f^{(d)}\left(x\right)$ for any interested point $\bar{x} \in \mathbb{R}$, consider the function $h(x):=f\left(\bar{x}+x\right)$. By \cref{lem:shifted-f}, one has
\begin{align}
    f^{(d)}\left(\bar{x}\right) =h^{(d)} (0)=\frac{1}{2} \sum_{i=1}^r b_i(\mathbf{x})\left(h\left(x_i\right)-h\left(-x_i\right)\right)
 =\frac{1}{2} \sum_{i=1}^r b_i(\mathbf{x}) \Bigl( f\left(\bar{x} + x_i\right)-f\left(\bar{x} - x_i\right) \Bigr),
\end{align}
where second equation holds because, upon expanding $h (x)$, we observe that $h$ is also a trigonometric polynomial like \cref{eq:trig}, sharing the same frequency set $\{\Omega_k\}$ as $f$; the only difference lies in the coefficients $a_k$ and $b_k$. As a result, the EPSR of \cref{eq-557} can be directly applied to $h$. Now, we have showed the odd argument in \cref{thm:main-psr}.

\subsubsection{Computation of the even order derivatives}

We next consider the even part  \cref{eq-even-part}. We are free to choose $r+1$ interpolation nodes $\mathbf{x}=\left[x_0, x_1, \cdots, x_r\right]^T \in \mathbb{R}^{r+1}$ as long as the $(r+1) \times (r+1)$ interpolation matrix
\begin{equation}\label{eq-Ae-full}
    A_e \equiv A_e(\mathbf{x}):=\left[\begin{array}{c}
    q^{(0)}\left(x_0\right)^{T} \\
    q^{(0)}\left(x_1\right)^{T} \\
    \vdots \\
    q^{(0)}\left(x_r\right)^{T}
    \end{array}\right]=\left[\begin{array}{ccccc}
    1 & \cos \left(\Omega_1 x_0\right) & \cos \left(\Omega_2 x_0\right) & \cdots & \cos \left(\Omega_r x_0\right) \\
    1 & \cos \left(\Omega_1 x_1\right) & \cos \left(\Omega_2 x_1\right) & \cdots & \cos \left(\Omega_r x_1\right) \\
    \vdots & \vdots & \vdots & \ddots&\vdots \\
    1 & \cos \left(\Omega_1 x_r\right) & \cos \left(\Omega_2 x_r\right) & \cdots & \cos \left(\Omega_r x_r\right)
    \end{array}\right]
\end{equation}
is nonsingular. Let data vector
$
    y_e \equiv y_e (\mathbf{x}): =\left[f_e\left (x_0\right), f_e\left (x_1\right), \cdots, f_e\left (x_r\right)\right]^{T} \in \mathbb{R}^{r+1}.
$
Then, by definitions of even part, we have $A_e z_e =y_e$ and thus $z_e =A_e ^{-1} y_e$. Once $z_e$ is obtained, from \cref{eq-reform-even}, we have
\begin{align}
    f_e ^{(d)}(x)
    =\langle q^{(d)}(x), A_e ^{-1} y_e \rangle
    =\langle(A_e ^{T})^{-1} q^{(d)}(x), y_e \rangle.
\end{align}
In particular, we focus on case that $d$ is even and $x=0$. By \cref{eq-odd-even-d-0}, we have
\begin{equation}
    f^{(d)}(0)
    =f_e ^{(d)}(0)
    =\left\langle b(\mathbf{x}), y_e(\mathbf{x})\right\rangle = \sum_{i=0}^{r} b_i(\mathbf{x}) f_e\left(x_i\right)    = \frac{1}{2} \sum_{i=0}^r b_i(\mathbf{x})\left(f\left(x_i\right)+f\left(-x_i\right)\right) ,
\end{equation}
with $b(\mathbf{x}):=[A_e(\mathbf{x})^T]^{-1} q^{(d)} \in \mathbb{R}^{r+1}$. The remaining steps are identical to those in the case of the odd pattern, ultimately leading the even argument in \cref{thm:main-psr}.

\begin{remark}
Compared to calculating odd-order derivatives, which requires $r$ evaluations of $f_o$, calculating even-order derivatives necessitates $r+1$ evaluations of $f_e$, where in general both $f_o$ and $f_e$ require two calls to $f$. Notably, we always have $f_e(0)=f(0)$. If we always set first node $x_0 \equiv 0$, we can reduce the number of evaluations of $f$ by one. Furthermore, if the frequencies $\Omega_k = k$ (\cref{assm-equi}), then $f$ is $2\pi$-periodic, and since $f(\pi) = f(-\pi)$, it follows that $f_e(\pi) = f(\pi)$. By always setting last node $x_r = \pi$, we can further reduce the number of evaluations of $f$ by one more. Ultimately, this results in $2r$ calls to $f$. 
\end{remark}

\section{Optimal shifts with minimal derivative variance}\label{sec:optimal}

In the previous discussions, we have always assumed that the cost function $f(x)$ can be evaluated accurately. 
However, notice that the cost function $f (x)$ in \cref{eq-cost-fun} is an expectation value, and in practice, due to the limited number of measurements (say, $N_x$ shots) on quantum devices, we can only estimate $f (x)$ within a certain range of statistical uncertainty. 
Specifically, we can only obtain the following random variable as an unbiased estimate of $f(x)$: 
\begin{equation}\label{eq_noise}
    \hat{f} (x) = f (x) + \hat{\epsilon},
\end{equation}
where $\hat{\epsilon}$ is a zero-mean noise term. According to the central limit theorem, the variance of $\hat{\epsilon}$ is given by
\begin{equation}
    \frac{\sigma^2(x)}{N_x},
\end{equation}
where $\sigma^2(x)$ represents the variance of measuring the value of $f (x)$ in a single shot, which depends on the specific quantum circuit and the observable in the problem\footnote{Indeed, $\sigma^2(x)=\langle\psi| U(x)^{\dagger} C^2 U(x)|\psi\rangle-[f(x)]^2$; however, the computational cost is prohibitive.}. The subscript $x$ in $N_x$ highlights the number of shots assigned specifically for estimating $f (x)$, indicating that different numbers of shots may be allocated to different values $f (x)$ for $x \in \mathbb{R}$ if needed.

Since most optimization algorithms (such as gradient descent, ADAM) for training PQCs only use the gradient information, we focus here on computing the first-order derivative of $f$. With the statistical estimate in \cref{eq_noise}, after applying \cref{thm:main-psr}, we obtain
\begin{equation}\label{eq_hatf_xbar}
    \hat{f}^{\prime} (\bar{x}) = \frac{1}{2} \sum_{i=1}^r b_i (\mathbf{x}) \left (\hat{f} (\bar{x} + x_i) - \hat{f} (\bar{x} - x_i) \right),
\end{equation}
where $b (\mathbf{x}) = \left[A_o (\mathbf{x})^T \right]^{-1} p^{ (1)}$ and $\mathbf{x} \in \mathbb{R}^r$ is arbitrary shift. We observe that $\hat{f}^{\prime} (\bar{x})$ serves as an unbiased estimator of the true derivative $f^{\prime} (\bar{x})$. This follows from the linearity of expectation, which gives
\begin{equation}
    \operatorname{E}[\hat{f}^{\prime} (\bar{x})] = f^{\prime} (\bar{x}).
\end{equation}

On the other hand, \cref{thm:main-psr} indicates that there is considerable flexibility in selecting the shifts $\mathbf{x} \in \mathbb{R}^r$. As long as the interpolation matrix $A_o (\mathbf{x})$ is invertible, $\hat{f}^{\prime} (\bar{x})$ will always be an unbiased estimate of the true derivative. A natural question arises: does the derivative variance $\operatorname{Var}[\hat{f}^{\prime} (\bar{x})]$ depend on the specific shifts $\mathbf{x}$? If so, what is the optimal shifts $\mathbf{x}^*$ to minimize the derivative variance? (A smaller variance increases the likelihood of obtaining the true derivative, thereby enhancing the effectiveness of any gradient-based optimization algorithm.) These questions are not straightforward to answer because we also need to consider the variance of a single shot and the allocation scheme for the total number of shots. We proceed to discuss these two points in detail. 

\subsection{Two shot allocation schemes}\label{sec-two-shots}

Now, suppose we are interested in computing the derivative at some point $\bar{x} \in \mathbb{R}$. Following the convention in existing studies \cite{mari2021estimating,wierichs2022general,markovich2024parameter}, we consider the following assumption, which is usually a good approximation in practice. 

\begin{assumption}[Constant variance]\label{ass-var}
The variance of measuring $f(\bar{x})$ is invariant under any shift $s$ in the parameter. In other words, for any $s \in \mathbb{R}$, the one-shot variance of observing $\hat{f}(\bar{x} + s)$ is identical to that of observing $\hat{f}(\bar{x})$, i.e., $\sigma^2(\bar{x}) = \sigma^2(\bar{x} + s) \approx \sigma^2$ for some unknown constant $\sigma^2$.
\end{assumption}

\begin{remark}
In \cref{app-error-bounds}, we characterize the discrepancy between the variance minimized under the constant-variance assumption and the true minimal variance in its absence. The resulting bounds (\cref{thm-appr-var} in \cref{app-error-bounds}) furnish qualified theoretical justification for employing the constant-variance assumption under our framework. However, throughout the main text, all variance-related discussions are conducted under \cref{ass-var}.
\end{remark}

For convenience, we rewrite $\hat{f}^{\prime} (\bar{x})$ in \cref{eq_hatf_xbar} as
\begin{equation}
    \hat{f}^{\prime} (\bar{x}) = \sum_{\mu=1}^{2r} \gamma_\mu \, \hat{f} \left (\bar{x} + \phi_\mu \right),
\end{equation}
with coefficients $\boldsymbol{\gamma}: = \frac{1}{2} \left[ b (\mathbf{x})^T, -b (\mathbf{x})^T \right]^T$ and shifts $\boldsymbol{\phi}: = \left[\mathbf{x}^T, -\mathbf{x}^T \right]^T$. Using \cref{ass-var} and the independence of the random variables $\hat{f} (x)$ corresponding to different $x \in \mathbb{R}$, we obtain
\begin{equation}\label{eq_final_var_0}
    \operatorname{Var}[\hat{f}^{\prime} (\bar{x})] = \sum_{\mu=1}^{2r} \left| \gamma_\mu \right|^2 \operatorname{Var} \left[ \hat{f} \left (\bar{x} + \phi_\mu \right) \right] = \sum_{\mu=1}^{2r} \left| \gamma_\mu \right|^2 \frac{\sigma^2}{N_\mu},
\end{equation}
where $N_\mu$ is the number of shots used to evaluate $\hat{f} \left (\bar{x} + \phi_\mu \right)$ for $\mu =1 , \ldots, 2r.$ To estimate the derivative $f^{\prime} (\bar{x})$, we are limited by the total number of shots allowed on a quantum computer, denoted as $N_{\mathrm{total}}$. (It is impossible to perform an infinite number of measurements.) We now consider two different shot allocation schemes.

\begin{description}[leftmargin=0pt,labelindent=0pt]
    \item[Uniform-shot scheme]  
     Let 
    \begin{equation}
        N_\mu = \frac{N_{\mathrm{total}}}{2r}
    \end{equation}
     for $\mu = 1, \dots, 2r$, meaning that the same number of shots is used for each $\hat{f} (\bar{x} + \phi_\mu)$. This is the simplest and most common approach. The derivative variance \cref{eq_final_var_0} in this case becomes
    \begin{equation}\label{eq-var-unif}
    \operatorname{Var}_{\mathrm{unif}}[\hat{f}^{\prime} (\bar{x})]
    =\sum_{\mu=1}^{2 r}\left|\gamma_\mu\right|^2 {\sigma}^2 \frac{2 r} {N_{\mathrm{total}}}
    =\frac{2 r{\sigma}^2 } {N_{\mathrm{total}}}\sum_{\mu=1}^{2 r}\left|\gamma_\mu\right|^2 
    =\frac{{\sigma}^2 } {N_{\mathrm{total}}} \, r \| b (\mathbf{x}) \|_2^2.
    \end{equation}

    \item[Weighted-shot scheme] 
    Let
    \begin{equation}
        N_\mu = N_{\mathrm{total}} \frac{\left| \gamma_\mu \right|}{\| \gamma \|_1}
    \end{equation}
    for $\mu = 1, \dots, 2r$, where $\| \gamma \|_1 = \sum_{\mu=1}^{2r} \left| \gamma_\mu \right|$. This scheme allocates more shots to evaluation terms $\hat{f} (\bar{x} + \phi_\mu)$ with larger coefficients $\left| \gamma_\mu \right|$, because the variance will scale with the square of $\left| \gamma_\mu \right|$. The derivative variance \cref{eq_final_var_0} in this case becomes
    \begin{equation}\label{eq-var-wgt}
        \operatorname{Var}_{\mathrm{wgt}}[\hat{f}^{\prime} (\bar{x})]
        = \sum_{\mu=1}^{2r} \left| \gamma_\mu \right|^2 \frac{\sigma^2}{N_{\mathrm{total}} \frac{\left| \gamma_\mu \right|}{\| \gamma \|_1}}
        = \frac{\sigma^2 \| \gamma \|_1^2}{N_{\mathrm{total}}}
        = \frac{\sigma^2}{N_{\mathrm{total}}} \, \| b (\mathbf{x}) \|_1^2.
    \end{equation}
\end{description}

Using the norm inequality $\|\mathbf{a}\|_1 \leq \sqrt{n} \cdot \|\mathbf{a}\|_2$ for any vector $\mathbf{a} \in \mathbb{R}^n$, we have $\| b (\mathbf{x}) \|_1^2 \leq r \| b (\mathbf{x}) \|_2^2.$
Thus, regardless of the choice of shifts $\mathbf{x} \in \mathbb{R}^r$, the derivative variance of weighted scheme is always smaller than that of uniform scheme. 
In fact, as shown in the following lemma, the weighted-shot scheme is optimal among all possible allocations.

\begin{lemma}[Optimal shot allocation scheme]\label{lem-opt-shot}
Fix $N_{\mathrm{total}}>0$ and coefficients $\gamma_\mu$ for $\mu=1, \ldots, 2 r$ (namely, $ b (\mathbf{x})$). Then, the weighted scheme $N_\mu = N_{\mathrm{total}} \frac{\left| \gamma_\mu \right|}{\| \gamma \|_1}$ exactly solves the optimization problem 
\begin{equation}
    \min _{N_\mu > 0} \operatorname{Var}[\hat{f}^{\prime} (\bar{x})]=\sum_{\mu=1}^{2 r}\left|\gamma_\mu\right|^2 \frac{\sigma^2}{N_\mu} \quad \text { s.t. } \quad \sum_{\mu=1}^{2 r} N_\mu=N_{\mathrm{total }} .
\end{equation}
\end{lemma}
\begin{proof}
Let $a_\mu:=\left|\gamma_\mu\right| \geq 0$. For any $N_\mu>0$ satisfying $\sum_{\mu=1}^{2 r} N_\mu=N_{\mathrm{total}}$, the Cauchy-Schwarz inequality implies
\begin{equation}
    \left(\sum_{\mu=1}^{2 r} \frac{a_\mu^2}{N_\mu}\right)\left(\sum_{\mu=1}^{2 r} N_\mu\right) \geq\left(\sum_{\mu=1}^{2 r} a_\mu\right)^2,
\end{equation}
and hence
\begin{equation}
    \sum_{\mu=1}^{2 r} \frac{a_\mu^2}{N_\mu} \geq \frac{\left(\sum_\mu a_\mu\right)^2}{\sum_\mu N_\mu}=\frac{\|\gamma\|_1^2}{N_{\mathrm{total}}} .
\end{equation}
Multiplying by the constant $\sigma^2$ yields $\operatorname{Var}[\hat{f}^{\prime}(\bar{x})]=\sigma^2 \sum_\mu \frac{a_\mu^2}{N_\mu} \geq \frac{\sigma^2}{N_{\mathrm{total}}}\|\gamma\|_1^2 .$ Note that Cauchy-Schwarz inequality is tight iff $\frac{a_\mu}{\sqrt{N_\mu}} \propto \sqrt{N_\mu}\Leftrightarrow  N_\mu \propto a_\mu=\left|\gamma_\mu\right| .$ Imposing $\sum_\mu N_\mu=N_{\mathrm{total}}$ yields $N_\mu=N_{\mathrm{total}} \frac{\left|\gamma_\mu\right|}{\|\gamma\|_1},$ exactly the weighted-shot scheme, achieving $\operatorname{Var}[\hat{f}^{\prime}(\bar{x})]=\frac{\sigma^2}{N_{\mathrm{total}}}\|\gamma\|_1^2.$
\end{proof}

In practice, a constant number of shots, say $N_{\mathrm{unif}}$, is typically chosen for evaluating any $\hat{f} (x)$, leading to $N_{\mathrm{total}} = 2rN_{\mathrm{unif}}$. 
This is especially common in code implementations, such as in Qiskit and similar frameworks, where a fixed \verb|shot| parameter is often specified.
In this case, by applying the weighted scheme to reallocate these $2rN_{\mathrm{unif}}$ shots, we can always further reduce the derivative variance. 

Finally, the above discussion naturally applies to the derivative variance of derivatives of any order $d \geq 2$. However, when $d$ is even, $b (\mathbf{x})$ is $(r + 1)$-dimensional, so $\operatorname{Var}_{\mathrm{unif}}[\hat{f}^{(d)} (\bar{x})] = \sigma^2 (r +1)\| b (\mathbf{x}) \|_2^2/N_{\mathrm{total}}$ when considering the uniform-shot scheme. 

\subsection{Minimization of the derivative variance}

After selecting the specific shot allocation scheme, the derivative variance $\operatorname{Var}[\hat{f}^{\prime} (\bar{x})]$ depends only on the nodes $\mathbf{x}$. Based on previous discussion, we have two different criteria. If we use the uniform scheme, to minimize the derivative variance in \cref{eq-var-unif}, we will solve
\begin{equation}\label{pro_P1}
    \tag{P-unif}
    \begin{aligned}
    \min_{\mathbf{x} \in \mathbb{R}^r} \quad F_{\mathrm{unif}} (\mathbf{x}) & \triangleq \frac{1}{2} \|b (\mathbf{x})\|_2^2
    =\frac{1}{2}\left\|\left[A_o(\mathrm{x})^T\right]^{-1} p^{(1)}\right\|_2^2
    = \frac{1}{2} \left\langle \left[ A_o (\mathbf{x})^T A_o (\mathbf{x}) \right]^{-1}, p^{ (1)} p^{ (1)^T} \right\rangle.
    \end{aligned}
\end{equation}
If we use the weighted scheme, the derivative variance is given in~\cref{eq-var-wgt}, and the problem becomes
\begin{equation}\label{pro_P2}
    \tag{P-wgt}
    \begin{aligned}
    \min_{\mathbf{x} \in \mathbb{R}^r} \quad F_{\mathrm{wgt}} (\mathbf{x}) & \triangleq \|b (\mathbf{x})\|_1.
    \end{aligned}
\end{equation}

To solve the above two problems, we need to know their gradients. Given the non-smoothness of $\ell_1$ norm, we will derive its subgradient for $F_{\mathrm{wgt}}$. The (sub)gradients of $F_{\mathrm{unif}}$ and $F_{\mathrm{wgt}}$ are provided by the following lemma, where we consider general odd $d$-order derivatives. Let sign function $\operatorname{sgn} (\cdot)$ defined as $\operatorname{sgn} (z) = 1$ if $z > 0$, $0$ if $z = 0$, and $-1$ if $z < 0$, applied component wise when input is a vector. Let $\operatorname{diag} (\cdot)$ denote the vector formed by the diagonal entries of the input matrix. 

\begin{lemma}[(Sub)gradients of $F_{\mathrm{unif}}$ and $F_{\mathrm{wgt}}$ for odd $d$-order derivatives]\label{lem-grad-odd}
Let $d$ be an arbitrary odd positive integer. Let the matrix $A_o^{(1)}(\mathbf{x})$ be defined as the component-wise derivative of $A_o(\mathbf{x})$ with respect to $x_i$:
\begin{equation}\label{eq-Ao1}
    A_o^{(1)}(\mathbf{x})
    :=\left[\begin{array}{c}
    p^{(1)}\left(x_1\right)^T \\
    p^{(1)}\left(x_2\right)^T \\
    \vdots \\
    p^{(1)}\left(x_r\right)^T
    \end{array}\right]
    =\left[\begin{array}{cccc}
    \Omega_1 \cos \left(\Omega_1 x_1\right) & \Omega_2 \cos \left(\Omega_2 x_1\right) & \cdots & \Omega_r \cos \left(\Omega_r x_1\right) \\
    \Omega_1 \cos \left(\Omega_1 x_2\right) & \Omega_2 \cos \left(\Omega_2 x_2\right) & \cdots & \Omega_r \cos \left(\Omega_r x_2\right) \\
    \vdots & \vdots & \ddots & \vdots \\
    \Omega_1 \cos \left(\Omega_1 x_r\right) & \Omega_2 \cos \left(\Omega_2 x_r\right) & \cdots & \Omega_r \cos \left(\Omega_r x_r\right)
    \end{array}\right].
\end{equation}
We have the following results:

\begin{enumerate}
    \item The gradient of
$F_{\mathrm{unif}} (\mathbf{x}) = \frac{1}{2}\left\|b (\mathbf{x})\right\|_2^2 = \frac{1}{2}\left\|\left[A_o (\mathbf{x})^T\right]^{-1} p^{ (d)}\right\|_2^2 
$
at $\mathbf{x} \in \mathbb{R}^r$ is given by
\begin{equation}
    \nabla F_{\mathrm{unif}} (\mathbf{x})
    = -\operatorname{diag}\left (A_o^{ (1)} (\mathbf{x}) A_o (\mathbf{x})^{-1} b (\mathbf{x}) b (\mathbf{x})^T\right).
\end{equation}
\item The subgradient of
$F_{\mathrm{wgt}} (\mathbf{x}) = \|b (\mathbf{x})\|_1 = \left\|\left[A_o (\mathbf{x})^T\right]^{-1} p^{ (d)}\right\|_1$
at $\mathbf{x} \in \mathbb{R}^r$, i.e., $\partial F_{\mathrm{wgt}} (\mathbf{x}) \subseteq \mathbb{R}^r$, satisfies
\begin{equation}
    \partial F_{\mathrm{wgt}} (\mathbf{x}) \ni-\operatorname{diag}\left (A_o^{ (1)} (\mathbf{x}) A_o (\mathbf{x})^{-1} \operatorname{sgn} (b (\mathbf{x})) b (\mathbf{x})^T\right).
\end{equation}
\end{enumerate}
\end{lemma}
The proof of the above lemma can be found in \cref{app-proof-grad}. 
For even $d$-order derivatives, similar results hold (see \cref{lem-grad-even} in \cref{app-proof-grad}), and the proofs follow analogously. 
Finally, if we wish to minimize the variance of the derivative estimates $\operatorname{Var}[\hat{f}^{\prime}(\bar{x})]$, for any given $\bar{x}$, the optimal shifts $\mathbf{x}^* \in \mathbb{R}^r$ need to satisfy $\nabla F_{\mathrm{unif}} (\mathbf{x}^*) = 0 ,$ or $\partial F_{\mathrm{wgt}} (\mathbf{x}^*) \ni 0.$ In general, obtaining an analytical solution is difficult. In practice, it is sufficient to solve it numerically using any (sub)gradient-based solver. Moreover, the expression of gradient of $F_{\mathrm{unif}} (\mathbf{x})$ can be used to verify whether a given $\mathbf{x}$ is the optimal solution to problem \cref{pro_P1}. 
In the following, we focus on analyzing the optimality of equidistant PSR nodes under the two shot schemes discussed above.

\subsection{Are equidistant PSR optimal?}\label{subsec:not-optimal}

In this subsection, we always suppose that \cref{assm-equi} holds.
In \cref{app-equ-are-special}, we have proved that the equidistant PSRs in \cref{eq:equidistant-shifts-1st,eq:equidistant-shifts-2rd} are special cases of EPSR.
We are curious whether the equidistant PSRs achieve the minimal variance for the criteria of \cref{pro_P1,pro_P2} under two different shot schemes, respectively.
The conclusion is: \textit{when using a uniform-shot scheme, it can be shown that equidistant PSRs are not optimal; however, when using a weighted-shot scheme, it can be proven that equidistant PSRs are globally optimal nodes!}
We now present the arguments for each case in turn. 

\subsubsection{Equidistant PSR are not optimal under uniform-shot scheme}

\begin{proposition}\label{prop-equ-is-not-optimal}
Let \cref{assm-equi} hold, i.e., $\Omega_k=k$. Consider order $d=1$. The equidistant PSR nodes, $x_i = \frac{\pi}{2r} + (i-1) \frac{\pi}{r}$ for $i = 1, 2, \ldots, r$, do not yield a zero gradient $\nabla F_{\mathrm{unif}} (\mathbf{x})$ and hence are not optimal solution for \cref{pro_P1}. In this case, $F_{\mathrm{unif}} (\mathbf{x})=\frac{2r^2+1}{6}$.
\end{proposition}

The proof of above proposition can be found in \cref{app-equ-optimal}.
\cref{prop-equ-is-not-optimal} raises a cautionary issue. Given the equidistant frequency, the equidistant PSR does not allow us to achieve the minimal derivative variance $\operatorname{Var}_{\mathrm{unif}}[\hat{f}^{\prime} (\bar{x})]$ when using a uniform-shot scheme. 
The equidistant PSR \cite{wierichs2022general} has been widely used, and it is common for users to subconsciously set a constant number of shots for each function evaluation. 
While this may not ultimately affect the convergence of various gradient-based algorithms for parameter training of PQCs, it still wastes quantum resources, as fewer shots could have yielded equally stable gradient estimates.
In conclusion, under the uniform-shot scheme, there exist shifts that outperform the equidistant PSR; however, these optimal shifts require solving problem \cref{pro_P1} numerically. 
By a similar line of reasoning, when order $d=2$, we also conclude that the equidistant PSR in \cref{eq:equidistant-shifts-2rd} cannot produce a zero gradient of $F_{\mathrm{unif}} (\mathbf{x})$ (see \cref{lem-grad-even} in \cref{app-proof-grad}). Indeed, this result holds for any $d$.

\subsubsection{Equidistant PSR are globally optimal under weighted-shot scheme}

The next theorem shows that when using a weighted-shot scheme, the equidistant PSRs are the globally optimal solutions of \cref{pro_P2}. The proof of next theorem can be found in \cref{app-global}.

\begin{theorem}[Global optimality of equidistant PSR]\label{thm-global}
Let \cref{assm-equi} hold, i.e., $\Omega_k=k$. Then,
\begin{enumerate}
    \item  For any odd integer $d\geq1$, the equidistant PSR nodes $x_i=\frac{\pi}{2 r}+(i-1) \frac{\pi}{r}, i=1, \ldots, r$, globally solve
\begin{equation}
\min _{\mathbf{x} \in \mathbb{R}^r} F_{\mathrm{wgt}}(\mathrm{x})=\|b(\mathrm{x})\|_1=\left\|\left[A_o(\mathrm{x})^T\right]^{-1} p^{(d)}\right\|_1
\end{equation}
and minimum value is $r^d$.
\item For any even integer $d\geq 2$, the equidistant PSR nodes $x_i=\frac{i \pi}{r}, i=0,1, \ldots, r$, globally solve
\begin{equation}
\min _{\mathbf{x} \in \mathbb{R}^r} F_{\mathrm{wgt}}(\mathrm{x})=\|b(\mathrm{x})\|_1=\left\|\left[A_e(\mathrm{x})^T\right]^{-1} q^{(d)}\right\|_1
\end{equation}
and minimum value is $r^d$.
\end{enumerate}
Notice that the feasible regions of above two optimization problems are given in \cref{thm-cond-nonsingular}.
\end{theorem}

We next present two experiments demonstrating that the equidistant PSR is indeed a global minimizer numerically. 

\paragraph{Test I: landscapes of $F_{\mathrm{wgt}}(\mathbf{x})$.}

We consider the case of the equidistant frequency set $\{1,2\}$ with size $r = 2$, and examine odd derivatives $d = 1, 3, 5$.
In these cases, the adjustable nodes are $\mathbf{x} = (x_0,x_1) \in \mathbb{R}^2$. 
We plot the landscapes of the objective function $F_{\mathrm{wgt}}(\mathbf{x})$ in \cref{pro_P2}, as shown in panels (a), (c), and (e) of \cref{fig:min-l1norm}.
It is evident that the equidistant PSR shifts $\mathbf{x} =\left (\frac{\pi}{4}, \frac{3\pi}{4} \right)$ correspond to the global minimizers of $F_{\mathrm{wgt}}(\mathbf{x})$.
For even orders $d = 2, 4, 6$, the adjustable shifts are $\mathbf{x} = (x_0, x_1, x_2) \in \mathbb{R}^3$. To facilitate visualization, we fix $x_0 \equiv 0$ and plot $F_{\mathrm{wgt}}(\mathbf{x})$ as a function of $x_1$ and $x_2$ in panels (b), (d), and (f) of \cref{fig:min-l1norm}. As in the odd order case, the equidistant PSR shifts $\left(0, \frac{\pi}{2}, \pi \right)$ are observed to be the global minimizers.
These results indicate that the optimality of equidistant PSR nodes is independent of the derivative order $d$, when classified by parity.

\paragraph{Test II: error between numerical global solution and equidistant PSR nodes.}

We now consider the combinations with sizes $r = 1$ to $8$ (frequency set $\{1,\ldots,r\}$) and orders $d = 1$ to $8$, categorized by parity.
We use the differential evolution (DE) \cite{storn1997differential} solver to obtain numerical solutions to \cref{pro_P2}. 
Since DE is a global optimization solver, we assume it yields an approximate global minimizer $\mathbf{x}_{\mathrm{DE}}$. As shown in \cref{fig:h}, each color block represents the maximum error defined as $\max_i |(\mathbf{x}_{\mathrm{DE}})_i- (\mathbf{x}_{\mathrm{equi}})_i|$, where $\mathbf{x}_{\mathrm{equi}}$ denotes equidistant PSR nodes. The results demonstrate that the equidistant PSRs are globally optimal across all tested configurations.

\begin{figure}[htbp]
  \centering
  \begin{subfigure}[b]{0.41\textwidth}
    \includegraphics[width=\linewidth]{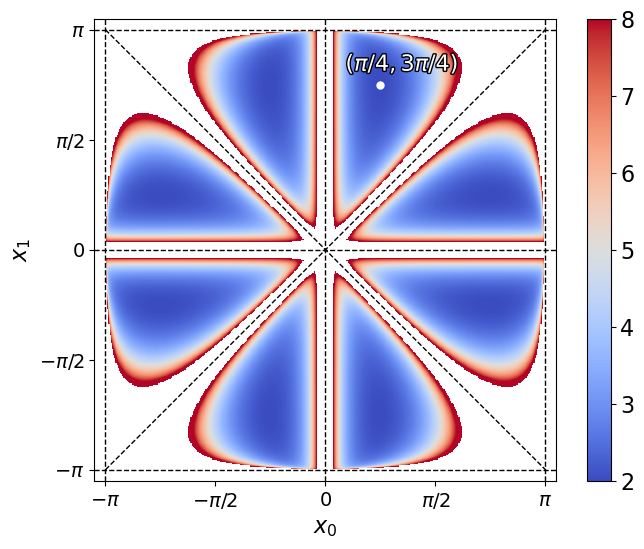}
    \caption{$d = 1$}
  \end{subfigure}
  \hspace{0.05\linewidth}
  \begin{subfigure}[b]{0.41\textwidth}
    \includegraphics[width=\linewidth]{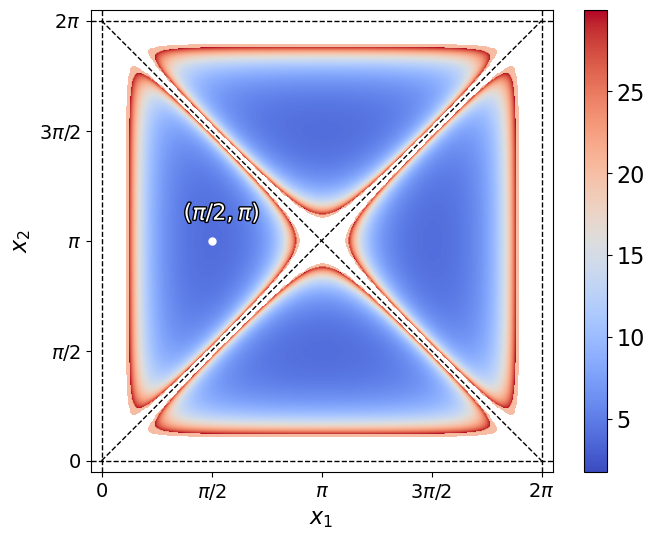}
    \caption{$d = 2$}
  \end{subfigure}

\vspace{0.5cm}

\begin{subfigure}[b]{0.41\textwidth}
    \includegraphics[width=\linewidth]{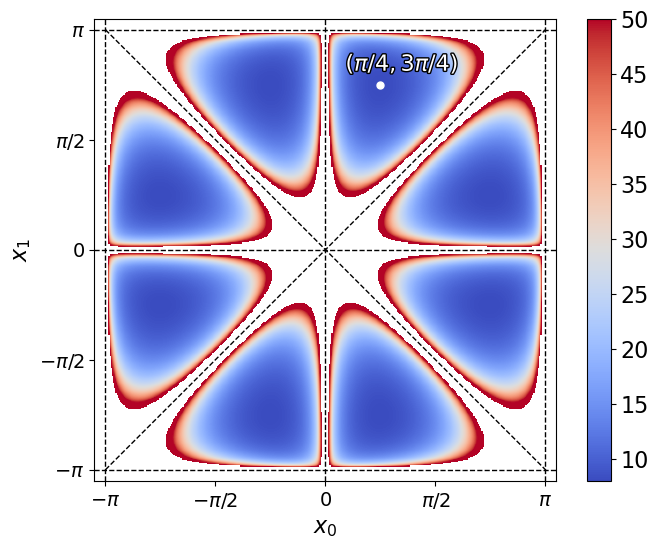}
    \caption{$d = 3$}
  \end{subfigure}
  \hspace{0.05\linewidth}
  \begin{subfigure}[b]{0.41\textwidth}
    \includegraphics[width=\linewidth]{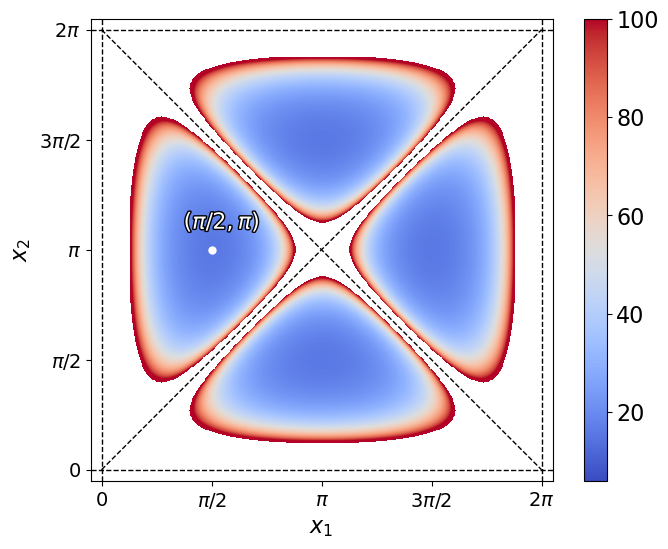}
    \caption{$d = 4$}
  \end{subfigure}

\vspace{0.5cm}

\begin{subfigure}[b]{0.41\textwidth}
    \includegraphics[width=\linewidth]{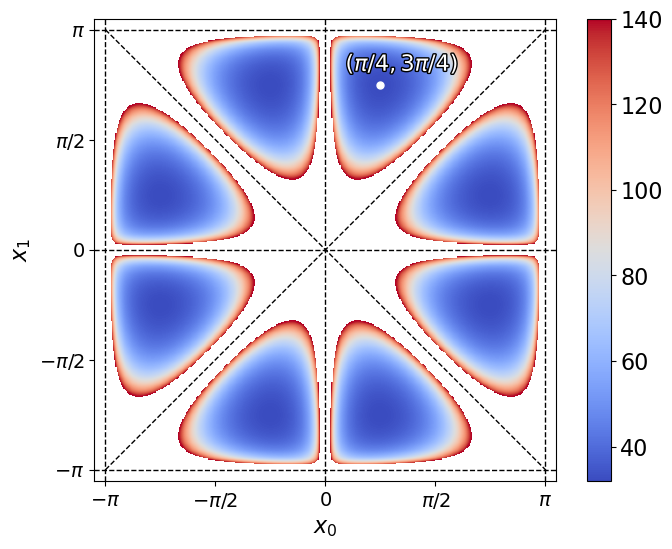}
    \caption{$d = 5$}
  \end{subfigure}
  \hspace{0.05\linewidth}
  \begin{subfigure}[b]{0.41\textwidth}
    \includegraphics[width=\linewidth]{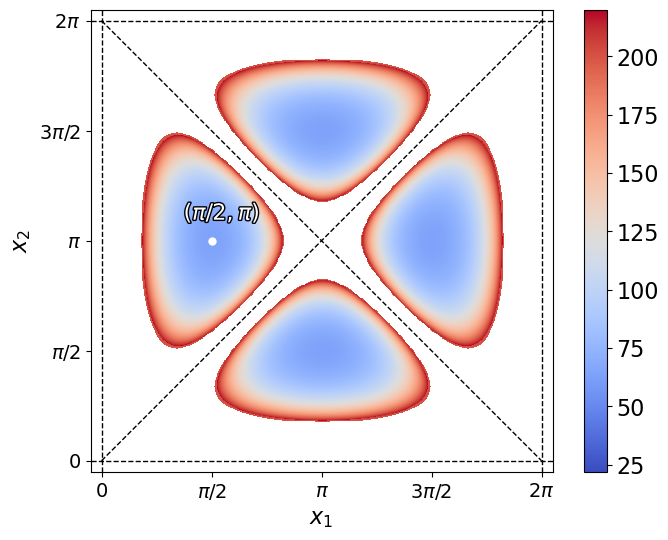}
    \caption{$d=6$}
  \end{subfigure}

\caption{Landscapes of $F_{\mathrm{wgt}}(\mathbf{x})$ for equidistant frequency set $\{1,2\}$ with $r=2$, across derivative orders $d=1$ to $6$. Panels (a), (c), and (e) show odd orders $d=1,3,5$ with $\mathbf{x} \in \mathbb{R}^2$, where the global minimizers are the equidistant PSR shifts $\left(\frac{\pi}{4}, \frac{3\pi}{4}\right)$. Panels (b), (d), and (f) show even orders $d=2,4,6$ with $\mathbf{x} \in \mathbb{R}^3$ (fixing $x_0 \equiv 0$), where the global minimizers are the equidistant PSR shifts $\left(0, \frac{\pi}{2}, \pi\right)$.}
\label{fig:min-l1norm}
\end{figure}

\begin{figure}[htbp]
    \centering
    \begin{subfigure}{0.3\linewidth}
        \includegraphics[width=\linewidth]{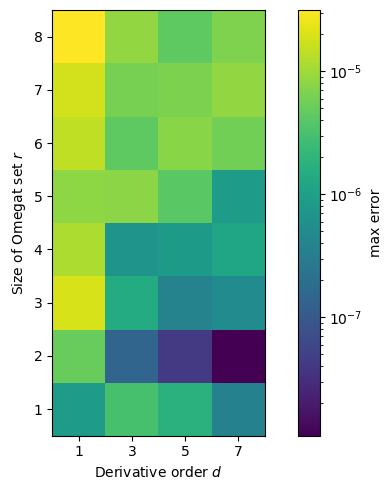}
    \end{subfigure}
    \hspace{0.05\linewidth}
    \begin{subfigure}{0.3\linewidth}
        \includegraphics[width=\linewidth]{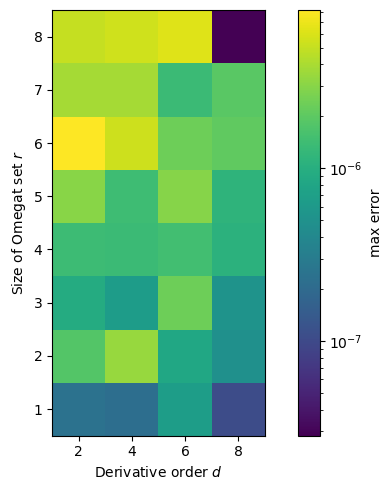}
    \end{subfigure}
    \caption{Maximum error between the numerical global minimizer obtained by differential evolution (DE), $\mathbf{x}_{\mathrm{DE}}$, and the equidistant PSR nodes, $\mathbf{x}_{\mathrm{equi}}$, for combinations of sizes $r = 1$ to $8$ and orders $d = 1$ to $8$, categorized by parity. Each color block indicates $\max_i |(\mathbf{x}_{\mathrm{DE}})_i - (\mathbf{x}_{\mathrm{equi}})_i|$.}
    \label{fig:h}
\end{figure}

\subsubsection{Conclusion of optimal shifts}

Finally, we conclude all results for the optimal shifts, where the criterion is the minimization of the derivative variance, making the estimated derivatives more robust to statistical noise.
\cref{tab:psr_comparison} summarizes the optimal shifts under different shot allocation schemes and frequency types.
It should be noted that the equidistant PSRs are only valid under the equidistant frequency set (\cref{assm-equi}), whereas the EPSR can handle non-equidistant frequencies, e.g., in the numerical experiments later, we will consider the frequency set $\{1,2,4\}$.
Importantly, for any chosen shifts $\mathbf{x}$, the derivative variance obtained by the weighted-shot scheme is always smaller than that of the uniform-shot scheme.
Therefore, regardless of the type of frequency set, one should always choose the weighted-shot. Further, if the frequencies are equidistant, the equidistant PSRs node should be selected; otherwise, one needs to numerically solve $\arg \min_{\mathbf{x}} F_{\mathrm{wgt}}(\mathrm{x}) =\|b(\mathrm{x})\|_1$ to get optimal node.

\begin{table}[htbp]
\centering
\renewcommand{\arraystretch}{1.2}
\setlength{\tabcolsep}{8pt}
\begin{tabular}{l p{6cm} p{6cm}}
\toprule
& \centering\textbf{Equidistant frequency $\Omega_k = k$} & \centering\textbf{Non-equidistant frequency $\Omega_k$} \tabularnewline
\midrule
\textbf{Weighted-shot scheme} &
\centering $\text{Equidistant PSR} = \displaystyle\arg\min_{\mathbf{x}} F_{\mathrm{wgt}}(\mathbf{x})$ \newline (\cref{thm-global}) &
\centering $\displaystyle\arg\min_{\mathbf{x}} F_{\mathrm{wgt}}(\mathbf{x})$ \tabularnewline
\addlinespace
\textbf{Uniform-shot scheme} &
\centering$\text{Equidistant PSR} \neq \displaystyle\arg\min_{\mathbf{x}} F_{\mathrm{unif}}(\mathbf{x})$ 
\newline (\cref{prop-equ-is-not-optimal}) &
\centering$\displaystyle\arg\min_{\mathbf{x}} F_{\mathrm{unif}}(\mathbf{x})$ \tabularnewline
\bottomrule
\end{tabular}
\caption{Optimal shifts $\mathbf{x}$ with minimal derivative variance under different shot allocation schemes and frequency types.}
\label{tab:psr_comparison}
\end{table}

\section{Numerical experiments}\label{sec-experiments}

In this section, we use the numerical simulations to validate the various conclusions drawn earlier. 
To implement the quantum circuits, we used the IBM Qiskit Aer Simulator \cite{qiskit}, which simulates the sampling process and the noisy environment of a real machine.
All tests were executed on a computer equipped with an AMD Ryzen 7 8845H CPU and 32 GB of RAM. 
The code is publicly available \cite{EPSRcode}.

\subsection{Problem setting}

In what follows, we consider a cost function derived from XXZ model with Hamiltonian variational ansatz (HVA) from \cite{wiersema2020exploring}.
Let $q$ denote the number of qubits. 
The symbol $\sigma_i^\alpha$ $(\alpha=x, y, z)$ corresponds to the Pauli operator acting on an $i$-th qubit. 
The Hamiltonian for XXZ model is given by
\begin{equation}
 H_{\mathrm{XXZ}} =H_{x x}+H_{y y}+\Delta H_{z z},
\end{equation}
where $H_{\alpha\alpha}=\sum_{i=1}^q \sigma_i^\alpha \sigma_{i+1}^\alpha$ $(\alpha=x, y, z)$ and $\Delta = 0.5$. Here, we use periodic boundary conditions, i.e., $\sigma_{q+1}^\alpha \equiv \sigma_1^\alpha$ $(\alpha=x, y, z)$. 
A depth-$p$ HVA circuit for the XXZ model corresponds to
\begin{equation}
\begin{aligned}
U(\boldsymbol{\theta}, \boldsymbol{\phi},\boldsymbol{\beta}, \boldsymbol{\gamma}) 
=  \prod_{l=1}^p G\left(\theta_l, H_{z z}^{\text {odd }}\right) G\left(\phi_l, H_{y y }^{\text {odd }}\right) G\left(\phi_l, H_{x x }^{\text {odd }}\right)
 G\left(\beta_l, H_{z z}^{\text {even }}\right) G\left(\gamma_l, H_{y y }^{\text {even }}\right) G\left(\gamma_l, H_{x x }^{\text {even }}\right),
\end{aligned}
\end{equation}
where $G(x, H)=\exp \left(-i \frac{x}{2} H \right)$, $H_{\alpha \alpha}^{\mathrm{even}}=\sum_{i=1}^{\lfloor q/2 \rfloor} \sigma_{2 i-1}^\alpha \sigma_{2 i}^\alpha $ and $H_{\alpha \alpha}^{\text {odd }}=\sum_{i=1}^{\lfloor q/2 \rfloor} \sigma_{2 i}^\alpha \sigma_{2 i+1}^\alpha$ for $\alpha=x, y, z$. Thus, there are $4 p$ parameters. 
\cref{fig:qc_XXZ_HVA} illustrates the HVA circuit for $q=5$ and $p=1$. 

In our experiment, we set $q=5$ and $p=2$. In this case,  we have $\boldsymbol{\theta} = (\theta_1, \phi_1, \beta_1, \gamma_1, \theta_2, \phi_2, \beta_2, \gamma_2) \in \mathbb{R}^{8}$. We observe that the parameters $\theta_1, \beta_1, \theta_2, \beta_2$ correspond to the equidistant frequency $\{1,2\}$ with $r=2$, while the parameters $\phi_1, \gamma_1, \phi_2$ correspond to the equidistant frequency $\{1,2,3,4\}$ with $r=4$. However, the last parameter $\gamma_2$ corresponds to the non-equidistant frequency $\{1,2,4\}$ with $r=3$. Note that for the last parameter $\gamma_2$, we cannot use equidistant PSRs.

\begin{figure}[htbp]
    \centering
\scalebox{0.8}{
\Qcircuit @C=1em @R=0.8em @!R { \\
\nghost{{q}_{0} :  } & \lstick{\ket{0}} &\gate{X} &\gate{H} &\ctrl{1} &\qw &\qw &\qw  &\qw &\qw &\qw  &\qw &\qw &\multigate{1}{R_{zz}(\beta)} &\qw &\multigate{1}{R_{YY}(\gamma)} &\multigate{1}{R_{XX}(\gamma)} &\qw &\qw \\
\nghost{{q}_{1} :  } & \lstick{\ket{0}} &\gate{X} &\qw &\targ &\qw &\qw &\multigate{1}{R_{zz}(\theta)}   &\qw &\multigate{1}{R_{YY}(\phi)}  &\multigate{1}{R_{XX}(\phi)} &\qw &\qw &\ghost{R_{zz}(\beta)} &\qw &\ghost{R_{YY}(\gamma)} &\ghost{R_{XX}(\gamma)} &\qw &\qw \\
\nghost{{q}_{2} :  } & \lstick{\ket{0}} &\gate{X} &\gate{H} &\ctrl{1} &\qw &\qw &\ghost{R_{zz}(\theta)} &\qw &\ghost{R_{YY}(\phi)}&\ghost{R_{XX}(\phi)}&\qw &\qw &\multigate{1}{R_{zz}(\beta)} &\qw &\multigate{1}{R_{YY}(\gamma)} &\multigate{1}{R_{XX}(\gamma)} &\qw &\qw \\
\nghost{{q}_{3} :  } & \lstick{\ket{0}} &\gate{X} &\qw &\targ &\qw &\qw &\multigate{1}{R_{zz}(\theta)}  &\qw &\multigate{1}{R_{YY}(\phi)} &\multigate{1}{R_{XX}(\phi)}  &\qw &\qw &\ghost{R_{zz}(\beta)} &\qw &\ghost{R_{YY}(\gamma)} &\ghost{R_{XX}(\gamma)} &\qw &\qw \\
\nghost{{q}_{4} :  } & \lstick{\ket{0}} &\gate{X} &\qw &\qw &\qw &\qw &\ghost{R_{zz}(\theta)}  &\qw &\ghost{R_{YY}(\phi)} &\ghost{R_{XX}(\phi)} &\qw &\qw &\qw &\qw &\qw &\qw &\qw &\qw  \gategroup{3}{8}{6}{8}{1.2em}{--}\gategroup{3}{10}{6}{11}{1.2em}{--}\gategroup{2}{14}{5}{14}{1.2em}{--}\gategroup{2}{16}{5}{17}{1.2em}{--}\\
\nghost{{q}_{5} :  } & & & & & &  &\mbox{$\theta$} & & \mbox{$\quad \quad \quad \quad \quad \quad \phi$}  & & & &\mbox{$\beta$} & & \mbox{$\quad\quad\quad\quad\quad \quad \gamma$} & & & &\\
\\ }
}
\caption{The HVA quantum circuit for the XXZ model with $q = 5$ qubits and $p = 1$ layer.}
\label{fig:qc_XXZ_HVA}
\end{figure}
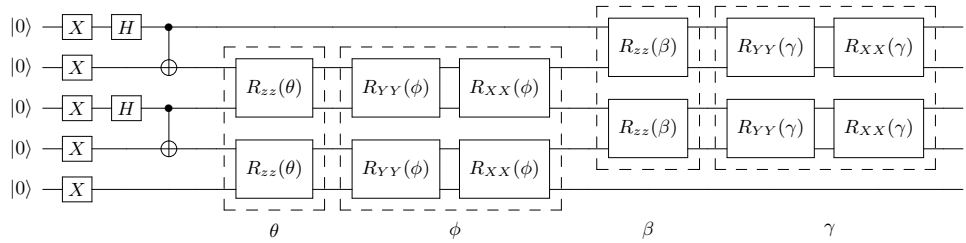

\subsection{Result I: validating the correctness of EPSR in a noise-free ideal setting}

To verify the correctness of EPSR in \cref{thm:main-psr}, we first consider a noise-free, ideal setting, where the exact value of the cost function in the XXZ model with HVA can be evaluated.
The Python package \verb|numdifftools| is used to compute derivatives of various orders, serving as the reference ground truth.
We randomly generate the parameter vector
$\boldsymbol{\theta} = (\theta_1, \phi_1, \beta_1, \gamma_1, \theta_2, \phi_2, \beta_2, \gamma_2) \in \mathbb{R}^8$.
For each parameter, we fix all others and compute partial derivatives of orders $d = 1$ to $6$ at the current value.
In each case, we apply EPSR with 10 randomly generated sets of shifts $\mathbf{x}$ and calculate the mean error relative to the \verb|numdifftools| reference.

The results are presented in \cref{fig:exp1}.
In the absence of function estimation noise, the derivatives obtained using any set of nodes are identical.
As expected, the error increases with larger $d$, primarily due to the accumulation of numerical errors in \verb|numdifftools| itself as the derivative order grows.
For the most commonly used cases, $d = 1$ and $d = 2$, the errors remain well within an acceptable range ($<10^{-9}$), confirming the practical reliability of EPSR.

\begin{figure}[htbp]
    \centering
    \begin{subfigure}[b]{0.45\linewidth}
        \includegraphics[width=\linewidth]{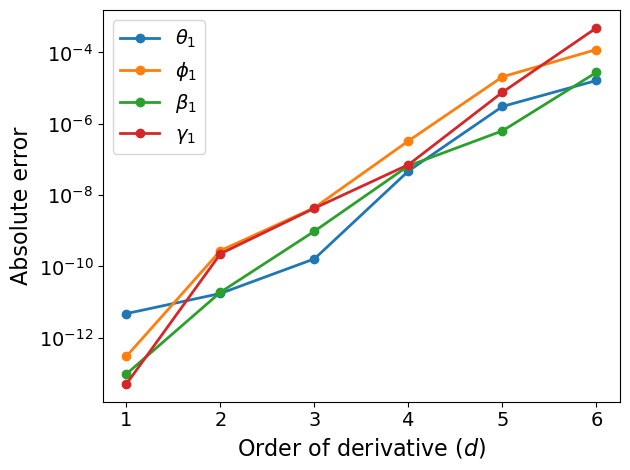}
    \end{subfigure}
    \hfill
    \begin{subfigure}[b]{0.45\linewidth}
        \includegraphics[width=\linewidth]{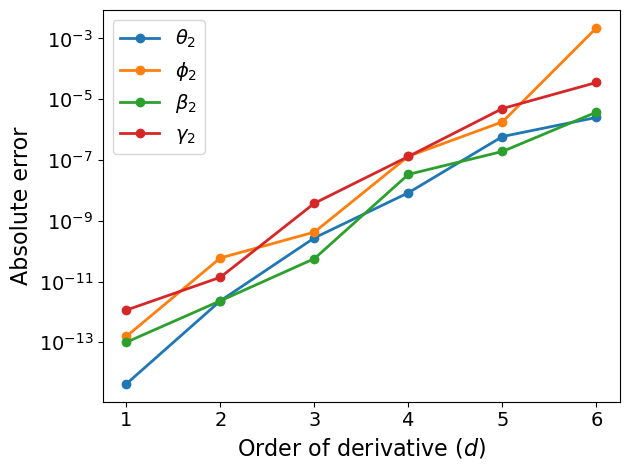}
    \end{subfigure}
    \caption{Error of EPSR derivatives compared to numdifftools derivatives for different orders in a noise-free ideal setting.}
    \label{fig:exp1}
\end{figure}

\subsection{Result II: weighted-shot scheme outperforms uniform-shot scheme}

In \cref{lem-opt-shot} of \cref{sec-two-shots}, we have theoretically proved that the weighted-shot scheme outperforms the uniform-shot scheme for any given shifts $\mathbf{x}$. In this experiment, we will show that weighted-shot indeed has smaller derivative variance in practice. 
We only consider the case where order $d=1$. 
Given any fixed total shooting counts of $N_{\mathrm{total}}>0$ and unknown constant $\sigma^2$, and based on \cref{eq-var-unif,eq-var-wgt}, we use the following scaled quantities as criteria for the derivative variance:
\begin{align}
\operatorname{\overline{Var}}_{\mathrm{unif}} & :=\frac{N_{\mathrm{total}}}{\sigma^2}    \operatorname{Var}_{\mathrm{unif}}[\hat{f}^{\prime} ]
    =  r\| b (\mathbf{x}) \|_2^2
    =  2 r F_{\mathrm{unif}}(\mathbf{x}),\\
\operatorname{\overline{Var}}_{\mathrm{wgt}}& :=\frac{N_{\mathrm{total}}}{\sigma^2}    \operatorname{Var}_{\mathrm{wgt}}[\hat{f}^{\prime} ]
    = \| b (\mathbf{x}) \|_1^2
    =[F_{\mathrm{wgt}}(\mathbf{x})]^2.   
\end{align}
Note that, in the following, the term \textit{derivative variance} refers to the quantities defined above.
Then, the variance ratio is defined by
\begin{equation}
    \text{Var}_{\text{ratio}} = \frac{\operatorname{\overline{Var}}_{\mathrm{unif}}}{\operatorname{\overline{Var}}_{\mathrm{wgt}}} = \frac{\frac{(2r^2 + 1)r}{3}}{r^2}= \frac{2r^2 + 1}{3r}.
\end{equation}
This ratio increases with larger values of $r$. 
In our XXZ model problem, we randomly choose a set of variables $\boldsymbol{\theta} = (\theta_1, \phi_1, \beta_1, \gamma_1, \theta_2, \phi_2, \beta_2, \gamma_2) \in \mathbb{R}^{8}$, and compute the partial derivatives with respect to the first two parameters, $\theta_1$ and $\phi_1$. 
Using equidistant PSR nodes and applying different shot allocation schemes (uniform vs. weighted), we set $N_{\mathrm{total}} = 1000$ and perform 500 sampling rounds. 
For each shot scheme, we estimate the probability density function (PDF) using Kernel density estimation and generate subplots.

For $\theta_1$ with $r = 2$, as shown in \cref{fig:exp2-1}, there is a slight difference between the weighted and uniform schemes, but the gap is not significant, with $\text{Var}_{\text{ratio}} = 1.5$. However, for $\phi_1$ with $r = 4$, where $\text{Var}_{\text{ratio}} = 2.75$, as seen in \cref{fig:exp2-2}, the difference becomes much more pronounced, clearly illustrating the advantage of the weighted scheme as $r$ increases. As can be seen from the figures, for a given total number of shots $N_{\mathrm{total}}$, the variance of the weighted scheme is clearly smaller than that of the uniform scheme.

\begin{figure}[htbp]
    \centering
    \begin{subfigure}[b]{0.49\linewidth}
        \includegraphics[width=\linewidth]{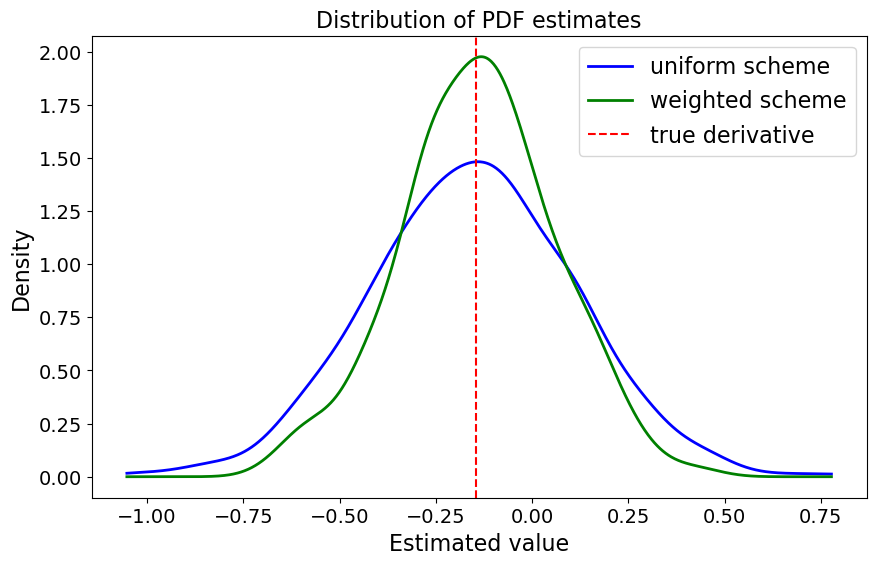}
        \caption{The variable $\theta_1$ with $r = 2$ has a derivative variance of 6 for the uniform-shot scheme, and a minimum value of 4 for the weighted-shot scheme.}
        \label{fig:exp2-1}
    \end{subfigure}
    \hfill
    \begin{subfigure}[b]{0.49\linewidth}
        \includegraphics[width=\linewidth]{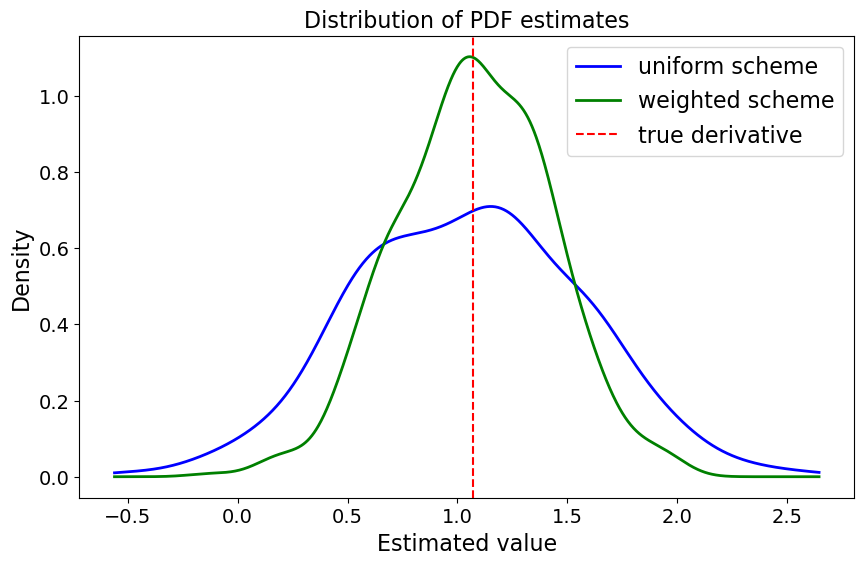}
        \caption{The variable $\phi_1$ with $r = 4$ has a derivative variance of 44 for the uniform-shot scheme, and a minimum value of 16 for the weighted-shot scheme.}
        \label{fig:exp2-2}
    \end{subfigure}
    \caption{Comparison of PDF distributions under different $r$ values and different shot allocation schemes (uniform vs. weighted) using equidistant PSR nodes.}
    \label{fig:exp2}
\end{figure}

\subsection{Result III: equidistant PSR nodes outperform random nodes under the weighted-shot scheme}

In this experiment, we focus on the weighted-shot scheme, and compare the performance of equidistant PSR nodes with that of various randomly chosen nodes. 
\cref{thm-global} has shown that equidistant PSR nodes are the optimal choice when using the weighted-shot scheme.
However, it is important to note that all our variance-related results are derived under the constant variance \cref{ass-var}, which does not hold in reality.
Therefore, we need to examine the practical performance to assess whether the conclusions drawn from this assumption remain reasonable in practice. 
Fortunately, the results are affirmative. 
(A more rigorous theoretical analysis is provided in \cref{app-error-bounds}.)
 
In our XXZ model problem, we randomly select a parameter vector $\boldsymbol{\theta} \in \mathbb{R}^8$ and compute the partial derivatives with respect to the first two parameters, $\theta_1$ and $\phi_1$. We set $N_{\mathrm{total}} = 1000$ and perform 500 sampling rounds. For each nodes $\mathbf{x}$, we estimate the probability density function (PDF) using Kernel density estimation and generate subplots.
The results are shown in the \cref{fig:exp3}. The results clearly show that equidistant PSR nodes consistently outperform random configurations by a large margin. This highlights the importance of node selection in EPSR. 

\begin{figure}[htbp]
    \centering
    \begin{subfigure}[b]{0.49\linewidth}
        \includegraphics[width=\linewidth]{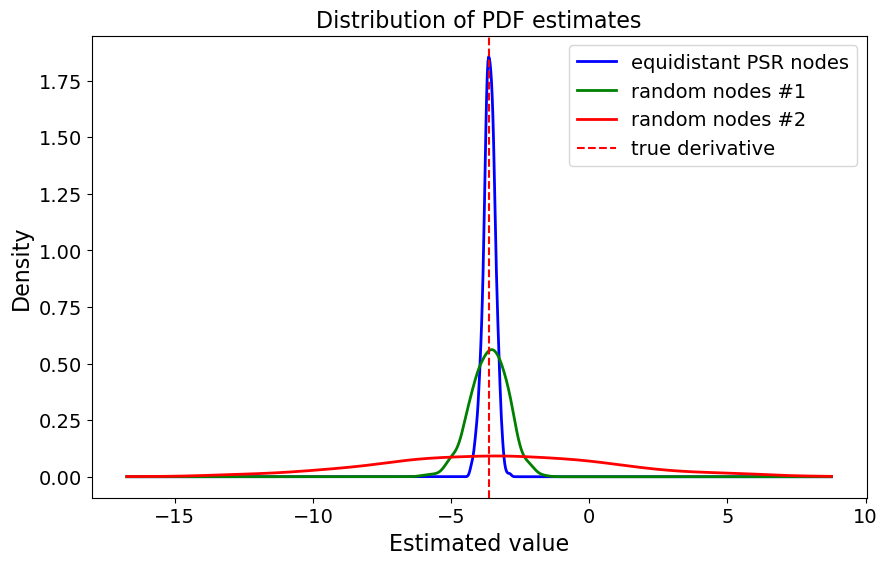}
        \caption{The variable $\theta_1$ with $r = 2$ has a derivative variance of approximately 31 for random nodes \#1, and 2,237 for random nodes \#2. The minimal derivative variance is 4 with equidistant PSR nodes.}
        \label{fig:exp3-1}
    \end{subfigure}
    \hfill
    \begin{subfigure}[b]{0.49\linewidth}
        \includegraphics[width=\linewidth]{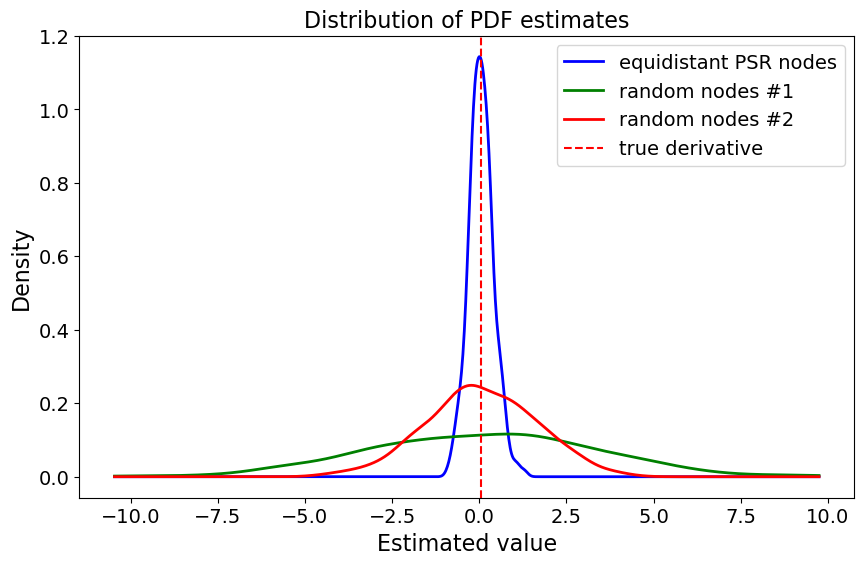}
        \caption{The variable $\phi_1$ with $r = 4$ has a derivative variance of approximately 1,674 for random nodes \#1, and 271 for random nodes \#2. The minimal derivative variance is 16 with equidistant PSR nodes.}
        \label{fig:exp3-2}
    \end{subfigure}
    \caption{Comparison of PDF distributions under different $r$ values and different nodes (equidistant PSR vs. random) using weighted-shot scheme.}
    \label{fig:exp3}
\end{figure}

\section{Discussion}\label{sec-discussion}

In this work, we introduce the extended parameter shift rule (EPSR), a unified framework for computing arbitrary-order derivatives in PQCs. Unlike traditional PSRs, which typically rely on fixed shift values, EPSR generalizes the approach by allowing infinitely many parameter shift choices and determining their corresponding coefficients through a tailored linear system. This flexibility enables the selection of shifts that minimize the variance of derivative estimates under constrained quantum resources.
Moreover, EPSR broadens the scope of existing PSRs by handling cases with non-equidistant generator spectra, which frequently arise beyond the standard Pauli operator setting. Despite this generalization, EPSR maintains a simple and transparent mathematical formulation without invoking complicated techniques. 
We also show that the commonly used equidistant PSRs can be recovered as a special case of EPSR and, under the weighted-shot allocation scheme, achieves globally optimal variance in theory.
Our numerical experiments using the Qiskit Aer simulator confirm the practical advantages of EPSR. 
In particular, choosing optimal parameter shifts significantly improves the accuracy of derivative estimation, highlighting the potential of EPSR as a powerful tool for enhancing gradient-based optimization. 
Below, we discuss several promising future works to conclude this paper.

\begin{enumerate}
    \item \textbf{Exploiting sparsity in trigonometric representation}: 
EPSR requires $2r$ function evaluations to compute a first-order derivative, where $r$ denotes the cardinality of the frequency set.
Specifically, consider a $q$-qubit system with a generator $H$ having $n$ unique eigenvalues $\{\lambda_l\}$ ($n \le 2^q$). 
Then, the cardinality $r$ of frequency set 
\begin{equation}\label{eq-4064}
    \{\Omega_k\} := \{ |\lambda_j - \lambda_l| > 0 \}
\end{equation}
is $\frac{n(n-1)}{2}$. Therefore, in the worst case, the number of function evaluations required by EPSR scales as $O(4^q)$.
In many practical scenarios, however, $r$ is significantly smaller.
For instance, in the XXZ model with HVA discussed in this work, $r = O(q)$; while in the TFIM model with HVA \cite{wiersema2020exploring,lai2025optimal}, $r = O(1)$.
This is because their $r$ depends jointly on the ansatz and the observable: in HVA, these two components are often highly correlated, and both play a critical role in determining the scale of $r$. 

In summary, the frequencies defined by \cref{eq-4064} only form a superset of the effective frequencies.
This definition is adopted merely for convenience, as it does not depend on the observable or the input state.
However, in many practical problems, the actual frequencies are jointly influenced by the generators of circuit, the observable, and the input state. For further details, see Appendices A in \cite{lai2025optimal}.
Future work could focus on identifying and exploiting the structural interplay between them to reduce the number of function evaluations required for EPSR.

    \item \textbf{Exploration of all possible linear combinations in PSR}: A potential concern is that PSR expresses the derivative of the PQC's cost function as a linear combination of function values. It would be valuable to investigate whether \cref{thm:main-psr} covers all possible such linear combinations, or if there are additional conditions or forms of linear combinations that need to be considered for a more general treatment of the problem. This exploration could help in broadening the scope and applicability of the EPSR framework.

\end{enumerate}


\begin{acknowledgments}
This work was supported in part by National Key Research and Development Program of China under the grant number 2024YFA1012903, the National Natural Science Foundation of China under the grant numbers 12501419, 12331010 and
12288101, and the Natural Science Foundation of Beijing, China under the grant number Z230002.
DA acknowledges the support by the Quantum Science and Technology-National Science and Technology Major Project via Project 2024ZD0301900, and the Fundamental Research Funds for the Central Universities, Peking University. 
Part of this work was completed while JH was affiliated with University of California, Berkeley.
\end{acknowledgments}

\section*{Data availability}
The data and codes used in this study are available from public repositories, see \cite{EPSRcode}.

\bibliography{mybib}

\appendix

\section{Equidistant PSRs are special cases of \cref{thm:main-psr}}\label{app-equ-are-special}

In this section, we will prove that the equidistant PSRs defined in \cref{eq-general-eq-psr-1st,eq-general-eq-psr-2rd}, as proposed in \cite{wierichs2022general}, are special cases of \cref{thm:main-psr}.

\subsection{First order derivatives}

To prove the equidistant PSR formula \cref{eq-general-eq-psr-1st} of first order derivative, we begin with two lemmas.

\begin{lemma}\label{lem-sum}
Given a positive integer $r$, define equidistant PSR nodes
\begin{equation}
    x_k=\frac{\pi}{2 r}+(k-1) \frac{\pi}{r}, \quad \text { for } k=1,2, \ldots, r,
\end{equation}
as in \cref{eq:equidistant-shifts-1st}. Let $m,n \in \mathbb{N}$. Then, the sum satisfies
\begin{equation}
    \sum_{k=1}^r \cos(m x_k ) =
    \begin{cases}
    r, & \text{if } m=2 n r \text{ for some even } n \in \mathbb{N}, \\
    -r, & \text{if } m=2 n r \text{ for some odd } n \in \mathbb{N}, \\
    0, & \text{if } m \neq 2 n r \text{ for any } n \in \mathbb{N}.
    \end{cases}
\end{equation}
\end{lemma}

\begin{proof}
Using the exponential representation of cosine, we rewrite the sum as  
\begin{align}
S 
:=\sum_{k=1}^r \cos \left(m x_k \right)
&=\frac{1}{2}\left[\sum_{k=1}^r\left(e^{i m}\right)^{x_k}+\sum_{k=1}^r\left(e^{-i m}\right)^{x_k}\right] \\
& = \frac{1}{2}\left[e^{-\frac{i m \pi}{2 r }} \sum_{k =1}^r\left(e^{i m \frac{\pi}{r }}\right)^k + e^{\frac{i m \pi}{2 r }} \sum_{k =1}^r\left(e^{-i m \frac{\pi}{r }}\right)^k\right]\\
& =\frac{1}{2}\left[p^{\frac{1}{2}} \cdot \sum_{k=1}^r q^k+q ^{\frac{1}{2}}\cdot \sum_{k=1}^r p^k\right], \label{eq-sum-S}
\end{align}
where in second equality, we substitute $x_k=\frac{\pi}{r}k -\frac{\pi}{2 r}$, and in last equality, we introduce the notations
\begin{equation}
    q := e^{i m \frac{\pi}{r} }, \quad p := e^{-i m \frac{\pi}{r} }.
\end{equation}
We now analyze $S$ by considering different cases for the integer $m$.

\textbf{Case 1}: $m=2 n r$ for some integer $n \in \mathbb{N}$. 
In this case, $q = e^{i m \frac{\pi}{r}}=e^{i 2 n \pi}=1$, and similarly, $p =1.$ 
Moreover, $q ^{\frac{1}{2}} = e^{\frac{i m \pi}{2 r}}=e^{i n \pi}=(-1)^{n}$, and similarly, $p^{\frac{1}{2}} =(-1)^{n}.$
Substituting these into  \cref{eq-sum-S}, we obtain $S=\frac{1}{2}\left[(-1)^n \cdot r+(-1)^n \cdot r\right]=(-1)^n \cdot r.$ Thus,
\begin{equation}
    S =
    \begin{cases}
    r, & \text{if } m = 2n r \text{ with } n \text{ even}, \\
    -r, & \text{if } m = 2n r \text{ with } n \text{ odd}. \\
    \end{cases}
\end{equation}

\textbf{Case 2}: $m \neq 2 n r$ for any integer $n \in \mathbb{N}$. 
In this case, $q \neq 1$ and $p \neq 1$ since $m \frac{\pi}{r} \neq 2 n \pi$ for any $n \in \mathbb{N}$.
Using \cref{eq-sum-S}, we derive
\begin{equation}
    S
    =\frac{1}{2}\left[p^{\frac{1}{2}} \cdot \frac{q\left(1-q^r \right)}{1-q}+q^{\frac{1}{2}} \cdot \frac{p\left(1-p^r \right)}{1-p}\right].
\end{equation}
Notice that $p q = 1$, and $q^r = e^{i m \pi}=(-1)^m,$ $p^r = e^{-i m \pi}=(-1)^m.$ If $m$ is even, then $q^r=p^r=1$, leading to $S=0$. If $m$ is odd, then $q^r=p^r=-1$, and we compute
\begin{align}
S
=\frac{1}{2}\left[p^{\frac{1}{2}} \cdot \frac{2q}{1-q}+q^{\frac{1}{2}} \cdot \frac{2p}{1-p}\right]
=\left[\frac{p^{\frac{1}{2}} q}{1-q}+ \frac{q^{\frac{1}{2}}pq}{(1-p)q}\right]
=\frac{1}{1-q}\left(p^{\frac{1}{2}} q-q^{\frac{1}{2}}\right).
\end{align}
Since $p q=1$, it follows that $p^{\frac{1}{2}} q-q^{\frac{1}{2}}=0$, thus $S=0.$ This completes the proof.
\end{proof}

The preceding lemma is used to establish the following result.

\begin{lemma}\label{lem-ATA}
We consider equidistant set $\Omega_k = k$ and choose equidistant shift nodes $x_k=\frac{\pi}{2 r}+(k-1) \frac{\pi}{r}$ as in \cref{eq:equidistant-shifts-1st} for $k=1,2, \ldots, r$, then the interpolation matrix $A_o(\mathbf{x})$ defined in \cref{eq-Ao} satisfies
\begin{equation}\label{eq-diag-D}
    A_o(\mathbf{x})^{T}A_o(\mathbf{x})
    =
    D:=
    r \left[\begin{array}{cccc}
    1 / 2 & & & \\
    & \ddots & & \\
    & & 1 / 2 & \\
    & & & 1
    \end{array}\right].
\end{equation}
\end{lemma}
\begin{proof}
Define the vectors $t_i=\left[\sin \left(i x_1\right), \sin \left(i x_2\right), \ldots, \sin \left(i x_r\right)\right]^T$ for $i=1,\ldots,r.$ Then, the matrix $A_o(\mathbf{x})=\left[t_1, t_2, \ldots, t_r\right]$. The $(i,j)$-th entry of the matrix $A_o(\mathbf{x})^{T} A_o(\mathbf{x})$ is given by: for all $i,j=1,\ldots,r,$
\begin{equation}\label{eq-sum-cos}
    t_i^{T} t_j
    =\sum_{k=1}^r \sin \left(i x_k\right) \sin \left(j x_k\right) \\
    =\frac{1}{2} \left(\sum_{k=1}^r \cos \left[ (i-j)x_k\right]-\sum_{k=1}^r \cos \left[(i+j)x_k\right]\right),
\end{equation}
where we use trigonometric identity $\sin a\sin b=\frac{1}{2}(\cos (a-b)-\cos (a+b))$.
For $i = j = r$, by \cref{lem-sum}, we have
\begin{equation}
    t_r^T t_r = \frac{1}{2} \left( r - \sum_{k=1}^r \cos(2r \cdot x_k) \right) = r.
\end{equation}
Similarly, for $i = j = 1, \ldots, r-1$, \cref{lem-sum} gives
\begin{equation}
    t_i^T t_i = \frac{1}{2} \left( r - \sum_{k=1}^r \cos(2i \cdot x_k) \right) = \frac{1}{2} r.
\end{equation}
Finally, in the case where $i \neq j$, the inner product $t_i^T t_j$ equals zero. This result follows because all possible values of $i + j$ are from the set $\{3, 4, 5, \ldots, 2r - 1\}$, and all possible values of $i - j$ belong to $\pm \{1, 2, \ldots, r - 1\}$. By \cref{lem-sum}, every summation term in \cref{eq-sum-cos} vanishes, leading to the desired conclusion.
\end{proof}

For $k=1,2, \ldots, r$, we consider the equidistant set $\Omega_k = k$ and select the equidistant PSR as $x_k=\frac{\pi}{2 r}+(k-1) \frac{\pi}{r}$. 
Thus, by \cref{lem-ATA}, we have $A_o(\mathbf{x})^{T} A_o(\mathbf{x}) = D$, where $D$ is a diagonal matrix given by \cref{eq-diag-D}. It follows that
\begin{equation}\label{eq-1050}
    [A_o(\mathbf{x})^{T}]^{-1} = A_o(\mathbf{x}) D^{-1},
\end{equation}
where
\begin{equation}
    A_o(\mathbf{x})
    =\left[\begin{array}{c}
    p^{(0)}\left(x_1\right)^T \\
    p^{(0)}\left(x_2\right)^T \\
    \vdots \\
    p^{(0)}\left(x_r\right)^T
    \end{array}\right]
    =
    \left[\begin{array}{cccc}
    \sin \left(x_1\right) & \sin \left(2 x_1\right) & \cdots & \sin \left(r x_1\right) \\
    \sin \left(x_2\right) & \sin \left(2 x_2\right) & \cdots & \sin \left(r x_2\right) \\
    \vdots & \vdots & \ddots& \vdots \\
    \sin \left(x_r\right) & \sin \left(2 x_r\right) & \cdots & \sin \left(r x_r\right)
    \end{array}\right].
\end{equation}
By \cref{thm:main-psr}, we obtain the coefficient vector $b =[A_o(\mathbf{x})^{T}]^{-1} p=A_o(\mathbf{x}) D^{-1} p$ with $p =\left[1, 2, \cdots, r\right]^T$. Then, the $i$-th entry of $b$ for $i=1, \ldots, r$ is
\begin{equation}\label{eq-1051}
    \begin{aligned}
    b_i
    & = p^{(0)}\left(x_i\right)^{T} (D^{-1} p)
    =\frac{1}{r} \left[\sum_{k=1}^{r-1} 2 k \sin \left(k x_i\right) + r \sin \left(r x_i\right) \right] \\
    & =\sin \left(r x_i\right)+\frac{2}{r} \sum_{k=1}^{r-1} k \sin \left(k x_i\right)
    = \frac{\sin \left(r x_i\right)}{2 r \sin ^2(\frac{1}{2} x_i)}.
    \end{aligned}
\end{equation}
Note that $\sin \left(r x_i\right) = \sin \left(\frac{\pi}{2 }+(i-1) \pi\right) = (-1)^{i-1}$ and we use $b^{\prime}=\frac{1}{2}[b(\mathbf{x})^T ;-b(\mathbf{x})^T]^T$ as final coefficients. So, this directly leads to the results in \cref{eq-general-eq-psr-1st}. Therefore, we showed the equidistant PSR \cite{wierichs2022general} for first derivative is a special case of our  \cref{thm:main-psr}.
Additionally, in this case, the $\ell_1$ norm of $b$ is
\begin{equation}\label{eq-d1-l1norm}
    \|b\|_1=\sum_{i=1}^r \frac{1}{2 r \sin ^2\left(\frac{(2 i-1) \pi}{4 r}\right)}=\frac{1}{2 r} \sum_{i=1}^r \csc ^2\left(\frac{(2 i-1) \pi}{4 r}\right) = r.
\end{equation}
\subsection{Second order derivatives}

To prove the equidistant PSR formula \cref{eq-general-eq-psr-2rd} of second derivatives, we begin with two lemmas. The overall approach in this subsection is essentially the same as in the previous one.

\begin{lemma}\label{lem-sum-2rd}
Given a positive integer $r$, define equidistant PSR nodes
\begin{equation}
    x_k = \frac{k\pi}{r}, \quad \text{for } k=0,1,\dots, r.
\end{equation}
as in \cref{eq:equidistant-shifts-2rd}. Let $m, n \in \mathbb{N}$. Then, the sum satisfies
\begin{equation}\label{eq-sum-cos-2rd}
    \sum_{k=0}^r \cos \left(m x_k\right)
    =
    \begin{cases}
    r+1, & \text {if } m=2 n r \text{ for some } n \in \mathbb{N}, \\
    1, & \text {if } m \neq 2 n r \text{ for any } n \in \mathbb{N}, \text{ and } m \text{ is even}, \\
    0, & \text {if } m \neq 2 n r \text{ for any } n \in \mathbb{N}, \text{ and } m \text{ is odd}.
    \end{cases}
\end{equation}
\end{lemma}
\begin{proof}
The proof process is essentially the same as that of \cref{lem-sum}.
\end{proof}

\begin{lemma}\label{lem-ATA-2rd}
We consider equidistant set $\Omega_k = k$ and choose equidistant shift nodes $x_k= \frac{k\pi}{r}$ as in \cref{eq:equidistant-shifts-2rd} for $k=0,1, \ldots, r$, then the interpolation matrix $A_e(\mathbf{x})$ defined in \cref{eq-Ae} satisfies
\begin{equation}\label{eq-nondiag-E}
    \forall i,j =0,1,\ldots, r, \quad E_{i j}
    :=\left[ A_o(\mathbf{x})^{T}A_o(\mathbf{x})\right]_{i j} = \begin{cases}
    r+1, & \text { if } i=j=0 \text { or } i=j=r, \\
    \frac{r+2}{2}, & \text { if } 0<i=j<r, \\
    1, & \text { if } i \neq j \text { and } i+j \text { is even,} \\
    0, & \text { if } i+j \text { is odd,}
    \end{cases}
\end{equation}
and the inverse of $E$ can be expressed as
\begin{equation}\label{eq-nondiag-invE}
    \forall i,j =0,1,\ldots, r, \quad [E^{-1}]_{i j}= \begin{cases}
    \frac{1}{d_i} \delta_{i j}-\frac{1}{2 d_i d_j}, & \text { if } i+ j \text { is even},
    \\ 0, & \text { if } i+j \text { is odd},
    \end{cases}
\end{equation}
where
\begin{equation}
    \forall i =0,1,\ldots, r, \quad d_i
    =
    \begin{cases}
    r, & \text { if } i=0 \text { or } i=r ,\\
    \frac{r}{2}, & \text { if } 1 \leq i \leq r-1.
    \end{cases}
\end{equation}
\end{lemma}
\begin{proof}
The proof process of \cref{eq-nondiag-E} is essentially the same as that of \cref{lem-ATA}, except that we use the result from \cref{lem-sum-2rd} here.  
The correctness of the inverse matrix $E^{-1}$ can be verified by inspection.
\end{proof}

For $k=0,1, \ldots, r$, we consider the equidistant set $\Omega_k = k$ and select the equidistant PSR as $x_k= \frac{k\pi}{r}$. 
Thus, by \cref{lem-ATA-2rd}, we have $A_e(\mathbf{x})^{T} A_e(\mathbf{x}) = E$. It follows that
\begin{equation}
    [A_e(\mathbf{x})^{T}]^{-1} = A_e(\mathbf{x}) E^{-1},
\end{equation}
where
\begin{equation}
    A_e(\mathbf{x})
    =\left[\begin{array}{c}
    q^{(0)}\left(x_0\right)^T \\
    q^{(0)}\left(x_1\right)^T \\
    \vdots \\
    q^{(0)}\left(x_r\right)^T
    \end{array}\right]
    =
    \begin{bmatrix}
    1 & \cos (x_0) & \cos (2 x_0) & \cdots & \cos (r x_0) \\
    1 & \cos (x_1) & \cos (2 x_1) & \cdots & \cos (r x_1) \\
    \vdots & \vdots & \vdots & \ddots & \vdots \\
    1 & \cos (x_r) & \cos (2 x_r) & \cdots & \cos (r x_r)
    \end{bmatrix}.
\end{equation}
By \cref{thm:main-psr}, we obtain the coefficient $b =A_e(\mathbf{x}) E^{-1} q$ with $q=-\left[0, 1^2,2^2 \cdots, r^2\right]^T$. Let $y:=E^{-1} q$, then
\begin{equation}
    \forall i =0,1,\ldots, r, \quad y_i=-\frac{i^2}{d_i}+\frac{1}{2 d_i} \sum_{\substack{j=0 \\ j \equiv i(\bmod 2)}}^r \frac{j^2}{d_j}, \quad \text{ with } d_i= \begin{cases}
    r, & i \in\{0, r\}, \\
    \frac{r}{2}, & 1 \leq i \leq r-1.
    \end{cases}
\end{equation}
For the $0$-th entry of $b$, since $x_0=0$, we have
\begin{equation}
    b_0
    = q^{(0)}\left(0\right)^{T} y
    =\sum_{i=0}^r y_i = -\frac{2 r^2+1}{6}.
\end{equation}
For the $i$-th entry of $b$ with $i=1,\ldots,r-1$, using $x_i= \frac{i\pi}{r}$, we have
\begin{equation}
    b_i
    = q^{(0)}\left(\frac{i\pi}{r}\right)^{T} y
    =\sum_{k=0}^r \cos \left(\frac{k i \pi}{r}\right) y_k=-\frac{2}{r} \sum_{k=1}^{r-1} k^2 \cos \left(\frac{k i \pi}{r}\right)-r \cos (i \pi)= \frac{(-1)^{i-1}}{\sin ^2\left(\frac{i \pi}{2 r}\right)}.
\end{equation}
For $r$-th entry of $b$, since $x_r=\pi$, we have
\begin{equation}
    b_r
    = q^{(0)}\left(\pi\right)^{T} y
    =\sum_{i=0}^r(-1)^i y_i= \frac{(-1)^{r+1}}{2}.
\end{equation}
Additionally, in this case, the $\ell_1$ norm of $b$ is
\begin{equation}\label{eq-d2-l1norm}
    \|b\|_1=\left|b_0\right|+\sum_{i=1}^{r-1}\left|b_i\right|+\left|b_r\right|=\frac{2 r^2+1}{6}+\frac{2\left(r^2-1\right)}{3}+\frac{1}{2}=r^2,
\end{equation}
where we use the fact that $\sum_{i=1}^{r-1} \frac{1}{\sin ^2\left(\frac{i \pi}{2 r}\right)}=\frac{2\left(r^2-1\right)}{3}$.

By rearranging these coefficients, we obtain \cref{eq-general-eq-psr-2rd}.
Hence, we showed the equidistant PSR \cite{wierichs2022general} for second derivative is a special case of our  \cref{thm:main-psr}.

\section{Proof of \cref{thm-cond-nonsingular}}\label{app-thm-cond}

For $\Omega_k=k$ with $k=1, \ldots, r$, consider the matrix
\begin{equation}
A_o(\mathbf{x})=\left[\begin{array}{cccc}
\sin \left(x_1\right) & \sin \left(2 x_1\right) & \cdots & \sin \left(r x_1\right) \\
\sin \left(x_2\right) & \sin \left(2 x_2\right) & \cdots & \sin \left(r x_2\right) \\
\vdots & \vdots & \ddots & \vdots \\
\sin \left(x_r\right) & \sin \left(2 x_r\right) & \cdots & \sin \left(r x_r\right)
\end{array}\right] .
\end{equation}
Our goal is to determine the conditions on $\mathbf{x}$ under which its determinant is nonzero.
To this end, we make use of a decomposition of $\sin (k x)$ in terms of Chebyshev polynomials. The key identity is
\begin{equation}\label{eq-9491}
    \sin (kx)=\sin x\; U_{k-1} (\cos x), \quad 1 \leq k\leq r, 
\end{equation}
where $U_n$ denotes the Chebyshev polynomial of the second kind \cite{mason2002chebyshev}, uniquely determined by the recurrence
\begin{equation}
    U_0 (t)=1, \quad U_1 (t)=2t, \quad U_{n+1} (t)=2t\, U_n (t)-U_{n-1} (t),
\end{equation}
with $\deg U_n=n$ and leading coefficient $2^n$. Substituting \cref{eq-9491} into $A_o (\mathbf x)$ and introducing the notation $z_i=\cos x_i$, we obtain
\begin{equation}
A_o(\mathbf{x})=\underbrace{\left[\begin{array}{cccc}
\sin x_1 & 0 & \cdots & 0 \\
0 & \sin x_2 & \cdots & 0 \\
\vdots & \vdots & \ddots & \vdots \\
0 & 0 & \cdots & \sin x_r
\end{array}\right]}_{=: D} \underbrace{\left[\begin{array}{cccc}
U_0\left(z_1\right) & U_1\left(z_1\right) & \cdots & U_{r-1}\left(z_1\right) \\
U_0\left(z_2\right) & U_1\left(z_2\right) & \cdots & U_{r-1}\left(z_2\right) \\
\vdots & \vdots & \ddots & \vdots \\
U_0\left(z_r\right) & U_1\left(z_r\right) & \cdots & U_{r-1}\left(z_r\right)
\end{array}\right]}_{=: B} .
\end{equation}
Next we factor $B$ as ``Vandermonde $\times$ upper triangular'' by expanding each Chebyshev polynomial $U_{k-1} (t)$ ($1 \leq k\leq r$) in the monomial basis:
\begin{equation}
    U_{k-1} (t)=\sum_{m=0}^{k-1} a_{m, k-1}\, t^m,
\end{equation}
with leading coefficient $a_{k-1, k-1}=2^{k-1}$. This yields
\begin{equation}\label{eq-mat-b}
B=\underbrace{\left[\begin{array}{ccccc}
1 & z_1 & z_1^2 & \cdots & z_1^{r-1} \\
1 & z_2 & z_2^2 & \cdots & z_2^{r-1} \\
\vdots & \vdots & \vdots & \ddots & \vdots \\
1 & z_r & z_r^2 & \cdots & z_r^{r-1}
\end{array}\right]}_{=: V} \underbrace{\left[\begin{array}{ccccc}
a_{0,0} & a_{0,1} & a_{0,2} & \cdots & a_{0, r-1} \\
0 & a_{1,1} & a_{1,2} & \cdots & a_{1, r-1} \\
0 & 0 & a_{2,2} & \cdots & a_{2, r-1} \\
\vdots & \vdots & \vdots & \ddots & \vdots \\
0 & 0 & 0 & \cdots & a_{r-1, r-1}
\end{array}\right]}_{=: T},
\end{equation}
where $V$ is the standard Vandermonde matrix and $T$ is upper triangular with diagonal entries $a_{k-1, k-1}=2^{ k-1}$ for $1 \leq k \leq r$. Finally, we have
\begin{equation}
    A_o (\mathbf x) = DVT.
\end{equation}
It follows immediately that
\begin{align}
    \det A_o (\mathbf x) &=\Big (\prod_{i=1}^r \sin x_i\Big) \Big (\prod_{1\le i<j\le r} (z_j-z_i)\Big) \Big (\prod_{k=1}^r 2^{k-1}\Big) \\
&=\Big (\prod_{i=1}^r \sin x_i\Big)
    \Big (\prod_{1\le i<j\le r}\big (\cos x_j-\cos x_i\big)\Big)2^{\frac{r (r-1)}{2}}.
\end{align}
We can directly read off the necessary and sufficient conditions from this factorization. Here is a quick check for
\begin{equation}
    A_o (\mathbf x)=\begin{bmatrix}\sin x_1&\sin 2x_1\\ \sin x_2&\sin 2x_2\end{bmatrix}.
\end{equation}
We compute $\det A_o (\mathbf x)= \sin x_1\sin 2x_2-\sin x_2\sin 2x_1=2\sin x_1\sin x_2\, (\cos x_2-\cos x_1),$ which matches the general formula for $r=2$.

Next, we consider the matrix
\begin{equation}
    A_e (\mathbf{x})=\left[\begin{array}{ccccc}
    1 & \cos \left (x_0\right) & \cos \left (2 x_0\right) & \cdots & \cos \left (r x_0\right) \\
    1 & \cos \left (x_1\right) & \cos \left (2 x_1\right) & \cdots & \cos \left (r x_1\right) \\
    \vdots & \vdots & \vdots & \ddots & \vdots \\
    1 & \cos \left (x_r\right) & \cos \left (2 x_r\right) & \cdots & \cos \left (r x_r\right)
    \end{array}\right].
\end{equation}
Here, the key identity is
\begin{equation}
    \cos (k x)=T_k (\cos x), \quad 0 \leq k\leq r,
\end{equation}
where $T_n$ denotes the Chebyshev polynomial of the first kind \cite{mason2002chebyshev}, uniquely determined by the recurrence
\begin{equation}
    T_0 (t)=1, \quad T_1 (t)=t, \quad T_{n+1} (x)=2 t \, T_n (t)-T_{n-1} (t),
\end{equation}
with $\operatorname{deg} T_n=n$ and leading coefficient $2^{n-1}$ for $n \geq 1$. Then, again setting $z_i=\cos x_i$, we have
\begin{equation}
    A_e (\mathbf{x})=\left[\begin{array}{ccccc}
    T_0\left (z_0\right) & T_1\left (z_0\right) & T_2\left (z_0\right) & \cdots & T_{r}\left (z_0\right) \\
    T_0\left (z_1\right) & T_1\left (z_1\right) & T_2\left (z_1\right)& \cdots & T_{r}\left (z_1\right) \\
    \vdots & \vdots & \vdots & \ddots & \vdots \\
    T_0\left (z_r\right) & T_1\left (z_r\right) & T_2\left (z_r\right)& \cdots & T_{r}\left (z_r\right)
    \end{array}\right].
\end{equation}
By analogy with the decomposition of $B$ in \cref{eq-mat-b}, we have
\begin{equation}
    \det A_e (\mathbf{x})=\Big (\prod_{0\le i<j\le r}\big (\cos x_j-\cos x_i\big)\Big)2^{\frac{r (r-1)}{2}}.
\end{equation}
For example, consider
\begin{equation}
    A_e (\mathbf{x})
    =\left[\begin{array}{lll}
    1 & \cos x_0 & \cos 2 x_0 \\
    1 & \cos x_1 & \cos 2 x_1\\
    1 & \cos x_2 & \cos 2 x_2
    \end{array}\right]
    =
    \left[\begin{array}{lll}
    1 & z_0 & 2 z_0^2-1 \\
    1 & z_1 & 2 z_1^2-1 \\
    1 & z_2 & 2 z_2^2-1
    \end{array}\right]
    =\left[\begin{array}{lll}
    1 & z_0 & z_0^2 \\
    1 & z_1 & z_1^2 \\
    1 & z_2 & z_2^2
    \end{array}\right]
    \left[\begin{array}{ccc}
    1 & 0 & -1 \\
    0 & 1 & 0 \\
    0 & 0 & 2
    \end{array}\right].
\end{equation}
Hence, $\det A_e (\mathbf{x})=2\left (z_1-z_0\right)\left (z_2-z_0\right)\left (z_2-z_1\right),$ which is in agreement with the general formula for $r=2$.

\section{(Sub)gradients of $F_{\mathrm{unif}}$ and $F_{\mathrm{wgt}}$}\label{app-proof-grad}

\subsection{Proof of  \cref{lem-grad-odd}}

\paragraph{Derivation of gradient of $F_{\mathrm{unif}}$.}
We now compute the gradient of
\begin{equation}
    F_{\mathrm{unif}}(\mathbf{x})=\frac{1}{2}\left\langle\left[A_o(\mathbf{x})^T A_o(\mathbf{x})\right]^{-1}, p^{(d)} p^{(d)^T}\right\rangle.
\end{equation}
We will define four sub-mappings $f_i$ ($i=1,2,3,4$) such that
$
    F_{\mathrm{unif}} = f_4 \circ f_3 \circ f_2 \circ f_1.
$
Next, we compute the differential formulas $D f_i$ at each point. By applying chain rule, we obtain
\begin{equation}
    D F_{\mathrm{unif}} : \mathbb{R}^{r} \to \mathbb{R}, \quad D F_{\mathrm{unif}} = D f_4 \circ D f_3 \circ D f_2 \circ D f_1.
\end{equation}
Finally, using the gradient identity
\begin{equation}
    D F_{\mathrm{unif}}(\mathbf{x})[u] = \langle u, \nabla F_{\mathrm{unif}}(\mathbf{x}) \rangle, \quad \forall u \in \mathbb{R}^{r},
\end{equation}
we derive an explicit expression for the gradient of $F_{\mathrm{unif}}$ at any given $\mathbf{x} \in \mathbb{R}^{r}$. We begin as follows.

\begin{enumerate}
    \item Consider the interpolation matrix map $f_1: \mathbb{R}^{r} \rightarrow \mathbb{R}^{r \times r}, \mathbf{x} \mapsto A_o(\mathbf{x})$. Its differential is given by
\begin{equation}
    D f_1(\mathbf{x})[u]=\operatorname{Diag}(u) A_o^{(1)}(\mathbf{x}), \quad \forall \, u \in \mathbb{R}^{r}.
\end{equation}
where $\operatorname{Diag}(u)$ denotes the diagonal matrix induced by vector $u$, and $A_o^{(1)}(\mathbf{x})$ is the $r \times r$ matrix whose $(i, j)$-entry is $\Omega_j \sin \left(\Omega_j x_i\right)$, see \cref{eq-Ao1}.
\item Consider the Gram matrix map $f_2: \mathbb{R}^{r \times r} \rightarrow \mathbb{R}^{r \times r}, X \mapsto X^T X$. Its differential at $X=A_o(\mathbf{x})$ is
\begin{equation}
    D f_2(X)[V_2]=V_2^T X+X^T V_2, \quad \forall \,V_2 \in \mathbb{R}^{r \times r}.
\end{equation}
\item Consider the matrix inversion map $f_3: \mathbb{R}^{r \times r} \rightarrow \mathbb{R}^{r \times r}, X \mapsto X^{-1}$. It is well known (see, e.g., \cite{horn2012matrix}) that its differential at $X=A_o(\mathbf{x})^T A_o(\mathbf{x})$ is
\begin{equation}
    D f_3(X)[V_3]=-X^{-1} V_3 X^{-1}, \quad \forall \,V_3 \in \mathbb{R}^{r \times r}.
\end{equation}
\item Consider the inner product $f_4: \mathbb{R}^{r \times r} \rightarrow \mathbb{R}, X \mapsto \frac{1}{2}\langle X, B\rangle$ with constant matrix $B:=p^{(d)} p^{(d)^T}$. Its differential at $X=\left[A_o(\mathbf{x})^T A_o(\mathbf{x})\right]^{-1}$ is
\begin{equation}
    D f_4(X)[V_4]=\frac{1}{2}\langle V_4, B\rangle, \quad \forall \,V_4 \in \mathbb{R}^{r \times r}.
\end{equation}
\end{enumerate}
Finally, the composition of all differentials $D f_i$ yields ($\forall u \in \mathbb{R}^{r}$)
\begin{equation}
    D F_{\mathrm{wgt}}(\mathbf{x})[u]=-\frac{1}{2}\left\langle\left(A^T A\right)^{-1} \left[ \left(U A^{(1)}\right)^T A + A^T\left(U A^{(1)}\right)\right]\left(A^T A\right)^{-1} ,B\right\rangle,
\end{equation}
where for simplicity, we use the notations $U=\operatorname{Diag}(u), A=A_{o}(\mathbf{x})$ and $A^{(1)}=A_{o}^{(1)}(\mathbf{x})$. We now extract $\nabla F_{\mathrm{unif}}(\mathbf{x})$ explicitly by rearrange the $D F_{\mathrm{unif}}(\mathbf{x})[u]$ as follows:
\begin{align}
D F_{\mathrm{unif}}(\mathbf{x})[u]
& = -\frac{1}{2}\left\langle A^{(1)^T} UA + A^T U A^{(1)},C\right\rangle, \\
& = -\frac{1}{2}\left\langle \operatorname{Diag}(u),A^{(1)} C A^T+A C A^{(1)^T}\right\rangle, \\
& = -\frac{1}{2}\left\langle u,\operatorname{diag}\left(A^{(1)} C A^T+A C A^{(1)^T}\right)\right\rangle,
\end{align}
where in the first equality, we introduce $C:=\left(A^T A\right)^{-1} B\left(A^T A\right)^{-1}$. From above, we read the gradient expression
\begin{equation}
    \nabla F_{\mathrm{unif}}(\mathbf{x}) = -\frac{1}{2} \operatorname{diag} \left( A^{(1)} C A^T + A C A^{(1)^T} \right)= -\operatorname{diag} \left( A^{(1)} C A^T\right).
\end{equation}
Moreover, substituting $B=p^{(d)} p^{(d)^T}$ and $C$ into the above expression, we obtain
\begin{align}
    \nabla F_{\mathrm{unif}}(\mathbf{x}) 
    & = -\operatorname{diag}\left(A^{(1)}\left(A^T A\right)^{-1} p^{(d)} p^{(d)^T}\left(A^T A\right)^{-1} A^T\right)\\
    & = -\operatorname{diag}\left(A^{(1)} A^{-1} b b^T\right),
\end{align}
where we use $b=\left(A^T\right)^{-1} p^{(d)}$. Thus, the gradient computation for $F_{\mathrm{unif}}(\mathbf{x})$ is complete.

\paragraph{Derivation of subgradient of $F_{\mathrm{wgt}}$.}
We now compute the subgradient of
\begin{equation}
    F_{\mathrm{wgt}}(\mathbf{x})=\|b(\mathbf{x})\|_1, \text{ with } b(\mathbf{x}) = \left[A_o(\mathbf{x})^T\right]^{-1} p^{(d)}.
\end{equation}
Let $g: \mathbb{R}^r \rightarrow \mathbb{R}^{r}, \mathbf{x} \mapsto b(\mathbf{x}) = \left[A_o(\mathbf{x})^T\right]^{-1} p^{(d)}$. Then, its differential is
\begin{equation}
D g(\mathbf{x})[u]=-\left[A_o(\mathbf{x})^T\right]^{-1} \left(\operatorname{Diag}(u) A_o^{(1)}(\mathbf{x})\right)^{T} \left[A_o(\mathbf{x})^T\right]^{-1}  p^{(d)}, \quad \forall\, u \in \mathbb{R}^r.
\end{equation}
Again, for simplicity, we use $U=\operatorname{Diag}(u), A=A_{o}(\mathbf{x})$, $A^{(1)}=A_{o}^{(1)}(\mathbf{x})$ and $b=b(\mathbf{x}) $, we have
\begin{align}
D g(\mathbf{x})[u] =-A_o^{-T} \left(UA_o^{(1)}\right)^{T} A_o^{-T}  p^{(d)}
 =-A_o^{-T} A_o^{(1)T}U b.
\end{align}
Finally, one subgradient of $F_{\mathrm{wgt}}(\mathbf{x})$ can be obtain as follows:
\begin{align}
\partial F_{\mathrm{wgt}}(\mathbf{x})[u] & =\left\langle \operatorname{sgn}(b) , D g(\mathbf{x})[u] \right\rangle  =\left\langle \operatorname{sgn}(b) , -A_o^{-T} A_o^{(1)T}U b\right\rangle \\
& =\left\langle -A_o^{(1)}A_o^{-1} \operatorname{sgn}(b) b^{T}, U \right\rangle =\left\langle -\operatorname{diag}\left(A_o^{(1)}A_o^{-1} \operatorname{sgn}(b) b^{T}\right), u \right\rangle.
\end{align}
Here, we use the fact that the \( \ell_1 \)-norm's subdifferential is the sign vector $\operatorname{sgn}(b)$ and the relationship between the differential of the composed functions to write the above expression. From above, we read one subgradient.

\subsection{(Sub)gradients of $F_{\mathrm{unif}}$ and $F_{\mathrm{wgt}}$ for even $d$-order derivatives}

\begin{lemma}[(Sub)gradients of $F_{\mathrm{unif}}$ and $F_{\mathrm{wgt}}$ for even $d$-order derivatives] \label{lem-grad-even}
Let $d$ be an arbitrary even positive integer. Let the matrix $A_e^{(1)}(\mathbf{x})$ be defined as the component-wise derivative of $A_e(\mathbf{x})$ with respect to $x_i$:
\begin{equation}\label{eq-Ae1}
    A_e^{(1)}(\mathbf{x}):=\left[\begin{array}{c}
    q^{(1)}\left(x_0\right)^T \\
    q^{(1)}\left(x_1\right)^T \\
    \vdots \\
    q^{(1)}\left(x_r\right)^T
    \end{array}\right]=-\left[\begin{array}{ccccc}
    0 & \Omega_1 \sin \left(\Omega_1 x_0\right) & \Omega_2 \sin \left(\Omega_2 x_0\right) & \cdots & \Omega_r \sin \left(\Omega_r x_0\right) \\
    0 & \Omega_1 \sin \left(\Omega_1 x_1\right) & \Omega_2 \sin \left(\Omega_2 x_1\right) & \cdots & \Omega_r \sin \left(\Omega_r x_1\right) \\
    \vdots & \vdots & \vdots & \ddots & \vdots \\
    0 & \Omega_1 \sin \left(\Omega_1 x_r\right) & \Omega_2 \sin \left(\Omega_2 x_r\right) & \cdots & \Omega_r \sin \left(\Omega_r x_r\right)
    \end{array}\right]
\end{equation}
We have the following results:

\begin{enumerate}
    \item The gradient of
$F_{\mathrm{unif}}(\mathbf{x}) = \frac{1}{2}\|b(\mathbf{x})\|_2^2 =\frac{1}{2}\left\|\left[A_e(\mathbf{x})^T\right]^{-1} q^{(d)}\right\|_2^2 
$ at $\mathbf{x} \in \mathbb{R}^{r+1}$ is given by
\begin{equation}
    \nabla F_{\mathrm{unif}} (\mathbf{x})
    = -\operatorname{diag}\left (A_e^{ (1)} (\mathbf{x}) A_e (\mathbf{x})^{-1} b (\mathbf{x}) b (\mathbf{x})^T\right).
\end{equation}
\item The subgradient of
$F_{\mathrm{wgt}} (\mathbf{x}) = \|b (\mathbf{x})\|_1 = \left\|\left[A_e (\mathbf{x})^T\right]^{-1} q^{ (d)}\right\|_1$
at $\mathbf{x} \in \mathbb{R}^{r+1}$, i.e., $\partial F_{\mathrm{wgt}} (\mathbf{x})\subseteq \mathbb{R}^{r+1}$, satisfies
\begin{equation}
   \partial F_{\mathrm{wgt}} (\mathbf{x}) \ni  -\operatorname{diag}\left (A_e^{ (1)} (\mathbf{x}) A_e (\mathbf{x})^{-1} \operatorname{sgn} (b (\mathbf{x})) b (\mathbf{x})^T\right).
\end{equation}
\end{enumerate}
\end{lemma}
\begin{proof}
The proof of \cref{lem-grad-even} follows the same reasoning as the proof of \cref{lem-grad-odd}.
\end{proof}

\section{Proof of \cref{prop-equ-is-not-optimal}}\label{app-equ-optimal}

\paragraph{To show that $ \nabla F_{\mathrm{unif}}(\mathbf{x}) $ is nonzero.}
Consider $\Omega_k = k$ and $d=1$. Let us compute the gradient of the function $F_{\mathrm{unif}}(\mathbf{x})$ at equidistant PSR $x_k = \frac{\pi}{2r} + (k-1) \frac{\pi}{r}$ for $k = 1, 2, \ldots, r$. By \cref{lem-grad-odd}, the gradient is given by
\begin{equation}
    \nabla F_{\mathrm{unif}}(\mathbf{x}) =-\operatorname{diag}\left(A_o^{(1)}(\mathbf{x}) A_o(\mathbf{x})^{-1} b(\mathbf{x}) b(\mathbf{x})^T\right)= -\operatorname{diag}\left(A_o^{(1)}(\mathbf{x}) D^{-1} p b(\mathbf{x})^T\right),
\end{equation}
where we use $b(\mathbf{x}) =[A_o(\mathbf{x})^T]^{-1} p=A_o(\mathbf{x}) D^{-1} p$ that proved in \cref{app-equ-are-special}. The matrix $D$ is given in \cref{eq-diag-D}.
Let $E := A_o^{(1)}(\mathbf{x}) D^{-1} p b(\mathbf{x})^T$, and consider the unit vector $e_i$, where $e_i \in \mathbb{R}^r$ is the standard basis vector with 1 in the $i$-th position and 0 elsewhere, for $i = 1, \ldots, r$. The diagonal element $E_{ii}$ is then given by
\begin{align}
E_{ii}
 = \langle e_i, E e_i \rangle = \langle e_i, A_o^{(1)}(\mathbf{x}) D^{-1} p b(\mathbf{x})^T e_i \rangle
 = b_i \langle A_o^{(1)}(\mathbf{x})^T e_i, D^{-1} p \rangle = b_i \langle p^{(1)}(x_i), D^{-1} p \rangle.
\end{align}
Expanding $E_{ii}$ further and substituting the expression for $b_i$, we have
\begin{align}
E_{i i} 
= b_i \left( \sum_{k=1}^{r-1} k \cos(k x_i) \cdot \frac{2k}{r} + r \cos(r x_i) \right) 
= \frac{(-1)^{i-1}}{r^2 \sin^2\left(\frac{1}{2} x_i\right)} \sum_{k=1}^{r-1} k^2 \cos \left(k x_i\right).
\end{align}
Thus, the $i$-th component of the gradient $\nabla F_{\mathrm{unif}}(\mathbf{x})$ at the equidistant PSR is
\begin{equation}
    \left (\nabla F_{\mathrm{unif}} (\mathbf{x})\right)_i
    = \frac{ (-1)^{i}}{r^2 \sin^2\left (\frac{1}{2} x_i\right)} \sum_{k=1}^{r-1} k^2 \cos \left (k x_i\right).
\end{equation}
This gradient is nonzero vector. To demonstrate this, we can consider the first element of the gradient:
\begin{equation}
    \frac{-1}{r^2 \sin^2\left(\frac{1}{2}x_1\right)} \sum_{k=1}^{r-1} k^2 \cos(k x_1) < 0.
\end{equation}
Since $k x_1 = \frac{\pi}{2} \cdot \frac{k}{r}$ for $k = 1, \ldots, r-1$, it follows that $0 < k x_1 < \frac{\pi}{2}$. Therefore, each term $\cos(k x_1) > 0$.

\paragraph{To show that $F_{\mathrm{unif}}(\mathbf{x})=\frac{2r^2+1}{6}$.}
Note that, in this case, $A_o(\mathbf{x})^T A_o(\mathbf{x})=D$ given in \cref{eq-diag-D}, and $p=[1,2, \cdots, r]^T$, thus,
\begin{align}
F_{\text {unif }}(\mathbf{x}) &=\frac{1}{2}\left\langle\left[A_o(\mathbf{x})^T A_o(\mathbf{x})\right]^{-1}, pp^{T}\right\rangle =\frac{1}{2}\left\langle D^{-1}, pp^{T}\right\rangle \\
&=\frac{1}{r} \left( \sum_{k=1}^{r-1} k ^2 + \frac{1}{2}r^2\right)=\frac{1}{r} \left( \frac{(r-1)r(2r-1)}{6} + \frac{1}{2}r^2\right) =\frac{2r^2+1}{6}.
\end{align}

\section{Proof of \cref{thm-global}}\label{app-global}

We only present the proof for the case where order $d$ is odd, the even case is analogous. Consider $\Omega_k = k$ and equidistant PSR nodes $x_i=\frac{\pi}{2r}+ (i-1)\frac{\pi}{r}$ $(i=1, \ldots, r).$ 

\paragraph{To show that $\|b\|_1 = r^{d}$.}

We first show that for any odd order $d$, it holds that $\|b\|_1 = r^{d}$. From $A_o (\mathbf{x})^{T} A_o (\mathbf{x}) = D = \mathrm{Diag}\left (\tfrac{r}{2}, \ldots, \tfrac{r}{2}, r\right)$ in \cref{lem-ATA}, it follows that
\begin{equation}
    b =[A_o (\mathbf{x})^{T}]^{-1} p^{ (d)}  = A_o (\mathbf{x}) D^{-1} p^{ (d)},
\end{equation}
where $p^{ (d)} = (-1)^{\frac{d-1}{2}}[1^{d}, 2^{d}, \ldots, r^{d}]^{T}$. Then, the $i$-th entry of $b$ is
\begin{equation}
    b_i = (-1)^{\frac{d-1}{2}}\! \left[\frac{2}{r}\sum_{k=1}^{r-1} k^{d} \sin (k x_i) + r^{d-1}\sin (r x_i)\right].
\end{equation}
Let $s: =[\sin (r x_1), \ldots, \sin (r x_r)]^{T}$. For the equidistant PSR nodes $x_i$, we have $\sin (r x_i) = (-1)^{i-1}$, and $s$ is precisely the $r$-th column of $A_o (\mathbf{x})$, i.e., $s = A_o (\mathbf{x}) e_r$ with $e_r=[0,\ldots,0,1]^T$.

For odd $d$, employing the standard decomposition of trigonometric sums (repeated applications of Abel summation), one can write
\begin{equation}
    \frac{2}{r}\sum_{k=1}^{r-1} k^{d}\sin (kt) + r^{d-1}\sin (rt)
     = \sin (rt) P_{r, d} (\cos t),
\end{equation}
where $P_{r, d} (\cdot)$ is some rational function that is positive for $t \in (0, \pi)$. Substituting $t = x_i$ (so that $\sin (rt)=\pm 1$) gives $\operatorname{sgn} (b_i) = (-1)^{\frac{d-1}{2}} \operatorname{sgn} (\sin (r x_i)).$ Hence, $\operatorname{sgn} (b) = (-1)^{\frac{d-1}{2}} \operatorname{sgn} (s)= (-1)^{\frac{d-1}{2}} s$, and therefore
\begin{equation}
    \|b\|_{1} = \operatorname{sgn} (b)^{T} b
    = (-1)^{\frac{d-1}{2}}\, s^{T} b.
\end{equation}
Then,
\begin{equation}
    \begin{aligned}
    s^{T} b
    &= (A_o e_r)^{T}\bigl (A_o D^{-1} p^{ (d)}\bigr)
    = e_r^{T} (A_o^{T} A_o) D^{-1} p^{ (d)}
    = e_r^{T} D D^{-1} p^{ (d)}
    = e_r^{T} p^{ (d)}
    = (-1)^{\frac{d-1}{2}} r^{d}.
    \end{aligned}
\end{equation}
Substituting back yields $\|b\|_{1} = r^{ d}.$ This result covers the case $d=1$ in \cref{eq-d1-l1norm}.

\paragraph{To show the global optimality.} 

Write $p: = p^{ (d)}$ for simplicity. Our main idea is to formulate $\min_{\mathbf{x} \in \mathbb{R}^r} F_{\mathrm{wgt}} (\mathbf{x})$ as a linear programming problem, namely
\begin{equation}\label{eq-primal}
    \min_{b\in\mathbb R^r}\, \|b\|_1\quad\text{s.t.}\quad A_o (\mathbf x)^Tb=p. \tag{$\mathrm{P}_{\mathbf{x}}$}
\end{equation}
Its dual problem can be readily written as
\begin{equation}\label{eq-dual}
    \max_{y\in\mathbb R^r}\, p^Ty\quad\text{s.t.}\quad \|A_o (\mathbf x)\, y\|_\infty\le 1, \tag{$\mathrm{D}_{\mathbf{x}}$}
\end{equation}
where infinity norm $\|\mathbf{a}\|_{\infty}:=\max _i\left|a_i\right|$. Here, we always assume that the chosen $\mathbf{x}$ makes $A (\mathbf{x})$ invertible; the condition for this is given in \cref{thm-cond-nonsingular}. Otherwise, $F_{\mathrm{wgt}} (\mathbf{x})$ is undefined. Therefore, for any fixed $\mathbf{x}$, the vector
\begin{equation}
    b (\mathbf{x}) =[A_o (\mathbf{x})^{T}]^{-1} p
\end{equation}
is the unique feasible solution to \cref{eq-primal} and hence its optimal solution, and $F_{\mathrm{wgt}} (\mathbf{x}) = \min \, (\mathrm{P}_{\mathbf{x}}) $.
We now seek a universal lower bound that holds for all $\mathbf{x}$. In particular, let
\begin{equation}
    y_{\mathrm{lb}}: = (-1)^{\frac{d-1}{2}} [0, \ldots, 0, 1]^{T}.
\end{equation}
Then, for any $\mathbf{x}$, $ \|A _o(\mathbf{x})\, y_{\mathrm{lb}}\|_{\infty} = \max_{i} \,\left| \sin (r x_{i}) \right| \le 1,$ which implies that $y_{\mathrm{lb}}$ is a feasible solution to \cref{eq-dual}. By weak duality,
\begin{equation}
    F_{\mathrm{wgt}} (\mathbf{x})=\|b(\mathrm{x})\|_1\ge  p^{T} y_{\mathrm{lb}}  =  r^{d}.
\end{equation}
This provides a lower bound that is independent of the nodes $\mathbf{x}$. As shown earlier, when $d$ is odd, and the equidistant PSR nodes $\mathbf{x}^*$ is chosen, then $F_{\mathrm{wgt}} (\mathbf{x}^*) = r^{d}.$ We complete the proof.

\section{Error analysis for minimal derivative variance beyond the constant-variance assumption}\label{app-error-bounds}

In \cref{sec:optimal}, we derived the minimal derivative variance under the constant-variance assumption (\cref{ass-var}). In practice, however, this assumption is seldom satisfied. Nonetheless, from a theoretical development perspective, it is challenging to make substantive progress without invoking it.
In this section, we analyze the difference between the minimal variance obtained under the constant-variance assumption and the true minimal variance without this assumption. 
The resulting error bounds provide a certain degree of theoretical support for using the constant-variance assumption in our framework.

Let us return to the beginning of \cref{sec:optimal}. Recall that our goal is to estimate the first derivative of $f$ at a point $\bar{x}$, denoted by $g:=f^{\prime}(\bar{x})$. After applying EPSR, this derivative can be expressed as a linear combination of several function evaluations:
\begin{equation}
    g=\sum_{\mu=1}^{2r} \gamma_\mu f\left (x_\mu\right),
\end{equation}
where the evaluation points are defined by $x_\mu:=\bar{x}+\phi_\mu$, and the shift vector is given by $\boldsymbol{\phi}: =\left[\mathbf{x}^T, -\mathbf{x}^T\right]^T$. The corresponding coefficients are $\boldsymbol{\gamma}: =\frac{1}{2}\left[b (\mathbf{x})^T, -b (\mathbf{x})^T\right]^T$.
On a quantum device, however, only noisy estimates of the function values are available. Specifically,
\begin{equation}
    \hat{f}\left (x_\mu\right)=f\left (x_\mu\right)+\hat{\varepsilon}_\mu, \qquad \mathbb{E}\left[\hat{\varepsilon}_\mu\right]=0 .
\end{equation}
If $N_\mu$ shots are allocated to the evaluation point $x_\mu$, then $\operatorname{Var}[\hat{f}\left (x_\mu\right)]= \sigma_\mu^2 / N_\mu,$ where $\sigma_\mu^2\equiv \sigma^2\left (x_\mu\right)$ denotes the single-shot variance. Consequently, the derivative estimator
\begin{equation}
    \hat{g}=\sum_{\mu=1}^{2r} \gamma_\mu \hat{f}\left (x_\mu\right), \qquad \operatorname{Var}[\hat{g}]=\sum_{\mu=1}^{2r} |\gamma_\mu|^2 \frac{\sigma_\mu^2}{N_\mu} .
\end{equation}
We note that $x_\mu, \gamma_\mu$, and $\sigma_\mu$ all depend on the chosen shift vector $\mathbf{x}$, although this dependence is omitted for notational convenience.

We first consider the optimal shot allocation for a fixed shift vector, without invoking the constant-variance assumption.
Given a fixed $\mathbf{x}$, our goal is to determine the optimal allocation of shots subject to a total budget constraint:
\begin{equation}
    \min _{N_\mu>0} \, \operatorname{Var}[\hat{g}]=\sum_{\mu=1}^{2r}  \gamma_\mu^2 \frac{\sigma_\mu^2}{N_\mu} \quad \text { s.t. } \quad \sum_{\mu=1}^{2r} N_\mu=N_{\mathrm{total}} .
\end{equation}
This is almost identical to the proof of \cref{lem-opt-shot}, except that $a_\mu:=\left|\gamma_\mu\right|$ is replaced here by $a_\mu :=\left|\gamma_\mu\right| \sigma_\mu$.
Applying the Cauchy–Schwarz inequality yields the optimal allocation scheme (referred to as the ``$\left|\gamma_\mu\right|\sigma_\mu$-weighted-shot'' scheme):
\begin{equation}\label{eq-opt-allocation}
    N_\mu^{\star} :=N_{\mathrm{total}} \frac{\left|\gamma_\mu\right| \sigma_\mu}{\sum_{\nu=1}^{2r} \left|\gamma_\nu\right| \sigma_\nu}, \quad \text { for } \mu=1, \ldots, 2 r,
\end{equation}
with corresponding minimal variance
\begin{equation}\label{eq-opt-cost}
    V_{\mathrm{opt}} :=\sum_{\mu=1}^{2r}  \gamma_\mu^2 \frac{\sigma_\mu^2}{N_\mu^{\star}}=\frac{1}{N_{\mathrm{total}}}\left (\sum_{\mu=1}^{2r} \left|\gamma_\mu\right| \sigma_\mu\right)^2 .
\end{equation}
\begin{remark}
When the constant-variance assumption holds (i.e., all $\sigma_\mu$ are equal), the expression above reduces to the weighted-shot scheme
\begin{equation}
    N_\mu^{\mathrm{w}}: =N_{\mathrm{total}} \frac{\left|\gamma_\mu\right|}{\sum_{\nu=1}^{2 r}\left|\gamma_\nu\right|}, \quad \text { for } \mu=1, \ldots, 2 r,
\end{equation}
as introduced in \cref{sec:optimal}; see \cref{eq-var-wgt}. To distinguish the two schemes, we refer to $N_\mu^{\mathrm{w}}$ as the ``$\left|\gamma_\mu\right|$-weighted-shot'' scheme. 
We further note that, since the variances $\sigma_\mu$ are unknown in practice, the truly optimal allocation $N_\mu^{\star}$ in \cref{eq-opt-allocation} is not accessible.
\end{remark}

In an idealized setting where the variance function $x \mapsto \sigma^2 (x)$ is fully accessible, one could further minimize \cref{eq-opt-cost} with respect to the shift vector:
\begin{equation}\label{eq-1835}
    \mathbf{x}_{\star}: =\arg \min _{\mathbf{x}} \, V_{\mathrm{opt}} (\mathbf{x})=\frac{1}{N_{\mathrm{total}}}\left (\sum_{\mu=1}^{2r}\left|\gamma_\mu (\mathbf{x})\right| \sigma_\mu (\mathbf{x})\right)^2,
\end{equation}
thereby obtaining the truly optimal choice of shift parameters.
However, because the quantities $\sigma_\mu(\mathbf{x})$ are unknown in practice, directly minimizing $V_{\mathrm{opt}}(\mathbf{x})$ is infeasible. To proceed, we adopt the constant-variance (CV) assumption, $\sigma_\mu(\mathbf{x}) \approx \sigma$, independent of both $\mu$ and $\mathbf{x}$. Under this approximation, \cref{eq-1835} reduces to
\begin{equation}\label{eq-cv-pro}
    \mathbf{x}_{\mathrm{CV}}: =\arg \min _{\mathbf{x}} F_{\mathrm{wgt}}(\mathbf{x}) \triangleq\sum_{\mu=1}^{2r} \left|\gamma_\mu (\mathbf{x})\right| =\|b(\mathbf{x})\|_1,
\end{equation}
which is precisely the approach taken in \cref{sec:optimal}.
In summary, our practical procedure consists of the following steps:
\begin{enumerate}
    \item Shift selection: Solve the surrogate problem \cref{eq-cv-pro} to obtain $\mathbf{x}_{\mathrm{CV}}$.
    \item Shot allocation: Employ the ``$\left|\gamma_\mu\right|$-weighted-shot'' scheme $N_\mu^{\text {w}}: =N_{\mathrm{total}} \frac{\left|\gamma_\mu\left (\mathbf{x}_{\mathrm{CV}}\right)\right|}{\sum_\nu\left|\gamma_\nu\left (\mathbf{x}_{\mathrm{CV}}\right)\right|}.$
\end{enumerate}
The resulting estimator has variance
\begin{equation}
    V_{\mathrm{alg}}: =\operatorname{Var}\left[\hat{g} \right]\,\text { using } \mathbf{x}_{\mathrm{CV}}, N_\mu^{\mathrm{w}}.
\end{equation}
We now compare this algorithmic variance $V_{\text{alg}}$ with the true global optimum $V_{\mathrm{opt}}\left(\mathbf{x}_{\star}\right)$ defined in \cref{eq-1835}, as stated in the following theorem.

\begin{theorem}\label{thm-appr-var}
If the single-shot variance satisfies $\sigma_{\min} \leq \sigma (x) \leq \sigma_{\max},$ with $\kappa: =\sigma_{\max} / \sigma_{\min},$ then the variance of the practical estimator $V_{\mathrm{alg}}$ obeys 
\begin{equation}
    1 \leq \frac{V_{\mathrm{alg}}}{V_{\mathrm{opt}}\left (\mathbf{x}_{\star}\right)} \leq \kappa^2.
\end{equation}
\end{theorem}

\cref{thm-appr-var} shows that, the practical strategy obtained by solving the surrogate problem~\eqref{eq-cv-pro} and using the $\lvert\gamma_\mu\rvert$-weighted-shot scheme remains provably near-optimal. The lower bound $V_{\mathrm{alg}} \geq V_{\mathrm{opt}}(\mathbf{x}_\star)$ is immediate since $V_{\mathrm{opt}}(\mathbf{x}_\star)$ is the global minimum over all shift vectors and shot allocations. The upper bound
\begin{equation}
    \frac{V_{\mathrm{alg}}}{V_{\mathrm{opt}} (\mathbf{x}_\star)} \leq \kappa^2 = \frac{\sigma^2_{\max}}{\sigma^2_{\min} }
\end{equation}
quantifies the penalty incurred by ignoring variation in the single-shot variance. When the measurement noise does not vary significantly across evaluation points (that is, when $\kappa \approx 1$), the variance achieved by our practical estimator is guaranteed to be close to the global optimum. Notably, the bound is independent of the total shot budget $N_{\mathrm{total}}$, the number of evaluation points $2r$, and the specific structure of the EPSR rule. This establishes broad theoretical support for employing the constant-variance assumption in practice.
The proof of \cref{thm-appr-var} is given below.

\begin{proof}[Proof of \cref{thm-appr-var}]
Define the quantities
\begin{equation}
    S (\mathbf{x}): =\sum_\mu\left|\gamma_\mu (\mathbf{x})\right|, \qquad w_\mu (\mathbf{x}): =\frac{\left|\gamma_\mu (\mathbf{x})\right|}{S (\mathbf{x})},
\end{equation}
so that $\left\{w_\mu\right\}$ forms a discrete probability distribution. Expectations with respect to this distribution are denoted by $\mathbb{E}_{\mathbf{x}}[\cdot]$. Under the $\lvert\gamma_\mu\rvert$-weighted-shot scheme (i.e., $N_\mu^{\mathrm{w}} := N_{\mathrm{total}} \frac{\lvert\gamma_\mu\rvert}{\sum_\nu \lvert\gamma_\nu\rvert}$), the true variance at a given shift vector $\mathbf{x}$ is\
\begin{equation}
    V_{\mathrm{w}} (\mathbf{x})
:=\sum_{\mu=1}^{2 r} \left|\gamma_\mu\right|^2\frac{ \sigma_\mu^2}{N_\mu^{\mathrm{w}}}
=\frac{\sum_\nu\left|\gamma_\nu\right|}{N_{\mathrm{total}}} \sum_{\mu=1}^{2 r}\left|\gamma_\mu\right| \sigma_\mu^2 
=\frac{S (\mathbf{x})^2}{N_{\mathrm{total}}} \mathbb{E}_{\mathbf{x}}\left[\sigma^2\right].
\end{equation}
In contrast, the variance of the ideal (optimal) $\lvert\gamma_\mu\rvert \sigma_\mu$-weighted scheme (see \cref{eq-opt-allocation,eq-opt-cost}) at the same shift $\mathbf{x}$ is
\begin{equation}
    V_{\mathrm{opt}} (\mathbf{x})
=\frac{1}{N_{\mathrm{total}}}\left(\sum_{\mu=1}^{2r}\left|\gamma_\mu \right| \sigma_\mu \right)^2
=\frac{S (\mathbf{x})^2}{N_{\mathrm{total}}}\left (\mathbb{E}_{\mathbf{x}}[\sigma]\right)^2.
\end{equation}
Therefore,
\begin{equation}
V_{\mathrm{alg}}=V_{\mathrm{w}}\left (\mathbf{x}_{\mathrm{CV}}\right)=\frac{S\left (\mathbf{x}_{\mathrm{CV}}\right)^2}{N_{\mathrm{total}}} \mathbb{E}_{\mathbf{x}_{\mathrm{CV}}}\left[\sigma^2\right],
\qquad
V_{\mathrm{opt}}\left (\mathbf{x}_{\star}\right)=\frac{S\left (\mathbf{x}_{\star}\right)^2}{N_{\mathrm{total}}}\left (\mathbb{E}_{\mathbf{x}_{\star}}[\sigma]\right)^2 .
\end{equation}
The ratio is thus
\begin{equation}
    \frac{V_{\mathrm{alg}}}{V_{\mathrm{opt}}\left (\mathbf{x}_{\star}\right)}=\frac{S\left (\mathbf{x}_{\mathrm{CV}}\right)^2}{S\left (\mathbf{x}_{\star}\right)^2} \cdot \frac{\mathbb{E}_{\mathbf{x}_{\mathrm{CV}}}\left[\sigma^2\right]}{\left (\mathbb{E}_{\mathbf{x}_{\star}}[\sigma]\right)^2} .
\end{equation}
Since $\mathbf{x}_{\mathrm{CV}}$ minimizes $S(\mathbf{x})$ under the surrogate problem~\eqref{eq-cv-pro}, we have $\frac{S\left (\mathbf{x}_{\mathrm{CV}}\right)^2}{S\left (\mathbf{x}_{\star}\right)^2} \leq 1 $. Moreover, from $\sigma_{\min } \leq \sigma_\mu \leq \sigma_{\max }$, it follows that $\mathbb{E}_{\mathbf{x}_{\mathrm{CV}}}\left[\sigma^2\right] \leq \sigma_{\max}^2$ and $\mathbb{E}_{\mathbf{x}_{\star}}[\sigma] \geq \sigma_{\min}$. Consequently,
\begin{equation}
    \frac{V_{\mathrm{alg}}}{V_{\mathrm{opt}}\left (\mathbf{x}_{\star}\right)} \leq \frac{\mathbb{E}_{\mathbf{x}_{\mathrm{CV}}}\left[\sigma^2\right]}{\left (\mathbb{E}_{\mathbf{x}_{\star}}[\sigma]\right)^2} \leq \frac{\sigma_{\max}^2}{\sigma_{\min}^2}=\kappa^2.
\end{equation}
For the lower bound, note first that for any shift vector $\mathbf{x}$, we have $ V_{\mathrm{w}} (\mathbf{x}) \geq V_{\mathrm{opt}} (\mathbf{x})$, since $V_{\mathrm{opt}} (\mathbf{x})$ is the minimum possible variance achievable by optimal shot allocation. Furthermore, because $\mathbf{x}_{\star}$ minimizes $V_{\text {opt}}(\mathbf{x})$, we have $V_{\mathrm{alg}}=V_{\mathrm{w}}\left (\mathbf{x}_{\mathrm{CV}}\right) \geq V_{\mathrm{opt}}\left (\mathbf{x}_{\mathrm{CV}}\right) \geq V_{\mathrm{opt}}\left (\mathbf{x}_{\star}\right).$
This completes the proof.
\end{proof}

\end{document}